\documentclass[fullpage,11pt]{article}

\usepackage{amsthm}
\usepackage{graphicx} 
\usepackage{array} 

\usepackage{amsmath, amssymb, amsfonts, verbatim}
\usepackage{hyphenat,epsfig,subfigure,multirow}

\usepackage[usenames,dvipsnames]{xcolor}
\usepackage[ruled]{algorithm2e}

\usepackage{tcolorbox}

\definecolor{DarkRed}{rgb}{0.5,0.1,0.1}
\definecolor{DarkBlue}{rgb}{0.1,0.1,0.5}

\usepackage{hyperref}
\hypersetup{
colorlinks=true,
pdfnewwindow=true,
citecolor=ForestGreen,
linkcolor=DarkRed,
filecolor=DarkRed,
urlcolor=DarkBlue
}

\usepackage{bm}
\usepackage{url}
\usepackage{xspace} 
\usepackage[mathscr]{euscript}

\usepackage{mdframed}

\usepackage[noend]{algpseudocode}
\makeatletter
\def\BState{\State\hskip-\ALG@thistlm}
\makeatother

\usepackage{cite}
\usepackage{enumerate}

\usepackage[margin=1in]{geometry}

\newtheorem{theorem}{Theorem}
\newtheorem{lemma}{Lemma}[section]
\newtheorem{proposition}[lemma]{Proposition}

\newtheorem{claim}[lemma]{Claim}

\newtheorem*{claim*}{Claim}
\newtheorem*{proposition*}{Proposition}
\newtheorem*{lemma*}{Lemma}
\newtheorem*{problem*}{Problem}

\newtheorem*{mdresult}{Main Result}

\allowdisplaybreaks

\renewcommand{\qed}{\nobreak \ifvmode \relax \else
      \ifdim\lastskip<1.5em \hskip-\lastskip
      \hskip1.5em plus0em minus0.5em \fi \nobreak
      \vrule height0.75em width0.5em depth0.25em\fi}

\usepackage[T1]{fontenc}
\usepackage[utf8]{inputenc}

\newcommand{\ourinfo}[1]{Department of Computer and Information Science, University of Pennsylvania. Supported in part by National Science Foundation grants CCF-1552909, CCF-1617851, and IIS-1447470.  \newline\noindent Email: \texttt{#1}.}

\newcommand{\eps}{\varepsilon}
\newcommand{\Paren}[1]{\Big(#1\Big)}
\newcommand{\Bracket}[1]{\Big[#1\Big]}
\newcommand{\bracket}[1]{\left[#1\right]}
\newcommand{\paren}[1]{\ensuremath{\left(#1\right)}\xspace}
\newcommand{\card}[1]{\left\vert{#1}\right\vert}

\newcommand{\floor}[1]{{\left\lfloor{#1}\right\rfloor}}
\newcommand{\prob}[1]{\Pr\paren{#1}}
\newcommand{\expect}[1]{\Ex{#1}}

\newcommand{\set}[1]{\ensuremath{\left\{ #1 \right\}}}

\newcommand{\OPT}{\ensuremath{\mbox{\sc opt}}\xspace}
\newcommand{\opt}{\ensuremath{\mbox{\sc opt}}\xspace}

\newcommand{\ALG}{\ensuremath{\mbox{\sc alg}}\xspace}

\DeclareMathOperator*{\Prob}{\ensuremath{\textnormal{Pr}}}
\renewcommand{\Pr}{\Prob}
\newcommand{\Ex}{\Exp}
\newcommand{\bigO}[1]{\ensuremath{O\Paren{#1}}}

\newcommand{\etal}{{\it et al.\,}}

\newcommand{\FG}{\ensuremath{\mathcal{G}}}

\newcommand{\event}[1]{\ensuremath{{\sf E}_{#1}}}

\newcommand{\rs}{Ruzsa-Szemer\'{e}di\xspace}

\newcommand{\bMatch}{\ensuremath{\floor{{1 \over p}}}}

\newenvironment{tbox}{\begin{tcolorbox}[
		enlarge top by=5pt,
		enlarge bottom by=5pt,
		 boxsep=0pt,
                  left=4pt,
                  right=4pt,
                  top=10pt,
                  arc=0pt,
                  boxrule=1pt,toprule=1pt,
                  colback=white
                  ]
	}
{\end{tcolorbox}}

\newcommand{\EX}[1]{\ensuremath{\mathbb{E}\Bracket{#1}}}
\renewcommand{\Ex}[1]{\ensuremath{\mathbb{E}[#1]}}
\newcommand{\constp}{\ensuremath{p_0}}
\newcommand{\basicalg}{\ensuremath{\textsf{MatchingCover}}\xspace}

\newcommand{\MP}{\ensuremath{\textnormal{MP-(\ref{eq:MP})}}\xspace}

\newcommand{\BinM}{\ensuremath{B_M}}
\newcommand{\BoutM}{\ensuremath{B_{\overline{M}}}}
\newcommand{\outMs}{\ensuremath{s_{\overline{M}}}}
\newcommand{\newM}{\ensuremath{M^+}}
\newcommand{\Fstar}{\ensuremath{F^{\star}}}

\title{The Stochastic Matching Problem: Beating Half with a Non-Adaptive Algorithm}  
\author{Sepehr Assadi\thanks{\ourinfo{\{sassadi,sanjeev,yangli2\}@cis.upenn.edu}}  \and Sanjeev Khanna\footnotemark[1] \and Yang Li\footnotemark[1]}

\date{}

\begin{document}
\maketitle

\thispagestyle{empty}
\begin{abstract}
  In the \emph{stochastic matching} problem, we are given a general (not necessarily bipartite) graph $G(V,E)$, where each edge in $E$ is
  \emph{realized} with some constant probability $p > 0$ and the goal is to compute a \emph{bounded-degree} (bounded by a function depending
  only on $p$) subgraph $H$ of $G$ such that the expected maximum matching size in $H$ is close to the expected maximum matching size in
  $G$. The algorithms in this setting are considered \emph{non-adaptive} as they have to choose the subgraph $H$ without knowing any
  information about the set of realized edges in $G$. Originally motivated by an application to \emph{kidney exchange}, the stochastic matching 
  problem and its variants have received significant attention in recent years.
    
  The state-of-the-art non-adaptive algorithms for stochastic matching achieve an approximation ratio of $\frac{1}{2}-\eps$ for any
  $\eps > 0$, naturally raising the question that if $1/2$ is the limit of what can be achieved with a non-adaptive algorithm.  In this
  work, we resolve this question by presenting the first algorithm for stochastic matching with an approximation guarantee that is strictly
  better than $1/2$: the algorithm computes a subgraph $H$ of $G$ with the maximum degree $\bigO{\frac{\log{(1/ p)}}{p}}$ such that
  the ratio of expected size of a maximum matching in realizations of $H$ and $G$ is at least $1/2+\delta_0$ for some absolute constant $\delta_0 > 0$. 
  The degree bound on $H$ achieved by our algorithm is essentially the best possible (up to an $O(\log{(1/p)})$ factor) for any
  \emph{constant factor} approximation algorithm, since an $\Omega(\frac{1}{p})$ degree in $H$ is necessary for a vertex to acquire at least
  one incident edge in a realization.  
  
  Our result makes progress towards answering an open problem of Blum~\etal (EC 2015) regarding the
  possibility of achieving a $(1-\eps)$-approximation for the stochastic matching problem using non-adaptive algorithms. From the technical
  point of view, a key ingredient of our algorithm is a structural result showing that a graph whose expected maximum matching size is
  $\opt$ always contains a $b$-matching of size (essentially) $b \cdot \opt$, for $b = \frac{1}{p}$.
\end{abstract}
\clearpage
\setcounter{page}{1}

\section{Introduction}\label{sec:intro}

We study the problem of finding a maximum matching in presence of \emph{uncertainty} in the input graph.  Specifically, we consider the
\emph{stochastic} setting where for an input graph $G(V,E)$ and a parameter $p > 0$, each edge in $E$ is realized \emph{independently}
w.p.\footnote{Throughout, we use \emph{w.p.}, \emph{w.h.p}, and \emph{prob.}  to abbreviate ``with probability'', ``with high probability'',
  and ``probability'', respectively.} $p$. We call the graph obtained from this stochastic process (which should be viewed as a random
variable) a \emph{realization} of $G(V,E)$, denoted by $G_p(V,E_p)$.  The \emph{stochastic matching} problem can now be defined as follows.
Given a general (not necessarily bipartite) graph $G(V,E)$ and an edge realization probability $p > 0$, compute a subgraph $H$ of $G$ such
that:

\begin{enumerate}[(i)]
\item The expected maximum matching size in a realization of $H$ is close to the expected maximum matching size in a realization of $G$.
\item The degree of each vertex in $H$ is bounded by some function that only depends on $p$, independent of the size of $G$.
\end{enumerate}
In other words, the stochastic matching problem asks if every graph $G$ contains a subgraph $H$ of \emph{bounded degree} (depending only on the realization probability $p$) 
such that the expected matching size in realizations of $G$ and $H$ are close.

\paragraph{Kidney exchange.} 
A canonical and arguably the most important application of the stochastic matching problem appears in \emph{kidney exchange}, where patients
waiting for kidney transplant can \emph{swap} their incompatible donors to each get a compatible donor. The goal is to identify a maximum
set of patient-donor pairs to perform such a swap (i.e., finds a maximum matching).  However, through medical records of patients and donors, one can
only filter out the patient-donor pairs where donation is \emph{impossible}, and more costly and time consuming tests must be performed
before a transplant can be performed.

The stochastic setting captures the essence of the need of extra tests for kidney exchange: an algorithm selects a set of patient-donor
pairs to perform the extra tests (i.e., computes a subgraph $H$), while making sure that there is a large matching among the pairs that pass
the extra tests. The objective that the subgraph $H$ has small degree captures the essence of minimizing the number of (costly and time
consuming) tests that each patient needs to go through.  The kidney exchange problem has been extensively studied in the literature,
particularly under stochastic settings (see, e.g.,~\cite{DickersonPS13,Anderson20012015,ManloveO14,Unver10,AkbarpourLG14,AndersonAGK15,AwasthiS09,DickersonPS12,DickersonS15}).
We remark that the the stochastic matching problem captures the simplest form of the kidney exchange, referred to as \emph{pairwise exchange}. Modern kidney exchange programs regularly 
employ swaps between three or patient-donor pairs and this setting has also been studied previously in the literature; we refer the interested reader to~\cite{BlumDHPSS15} for more details.


\paragraph{Previous work.} Our results are directly related to the results in~\cite{BlumDHPSS15} and~\cite{ECpaper} which we describe in
detail below.  Blum~\etal~\cite{BlumDHPSS15} introduced the (variant of) stochastic matching problem and proposed a
$(\frac{1}{2}-\eps)$-approximation algorithm (for any $\eps > 0$) which requires the subgraph $H$ to have maximum degree of
$\frac{\log{(1/\eps)}}{p^{\Theta(1/\eps)}}$. The algorithm of Blum~\etal~\cite{BlumDHPSS15} works as follows: Pick a maximum matching $M_i$
in $G$ and remove the edges in $M_i$; repeat for $R:=\frac{\log{(1/\eps)}}{p^{\Theta(1/\eps)}}$ times.  In order to analyze this algorithm,
the authors showed that, for any $i\in[R]$, if the size of the maximum matching among the realized edges in $M_1,\ldots,M_i$ is less than
$\opt/2$, the matching $M_{i+1}$ contains many augmenting paths of $M$ of length $O(\frac{1}{\eps})$; since each such augmenting path is
realized w.p. $p^{O(\frac{1}{\eps})}$, one needs to repeat this augmentation process for $\frac{1}{p^{O(\frac{1}{\eps})}}$ time (as is
roughly the value of $R$) to increase the matching size to $(\frac{1}{2}-\eps)\cdot\opt$.

In a recent work~\cite{ECpaper}, we showed that in order to obtain a $(\frac{1}{2}-\eps)$-approximation algorithm, one only needs a subgraph
$H$ with max-degree of $O(\frac{\log{(1/\eps p)}}{\eps p})$, significantly smaller than the bounds in~\cite{BlumDHPSS15}. Interestingly, the
algorithm of~\cite{ECpaper} and the one in~\cite{BlumDHPSS15} are essentially identical (modulo an extra sparsification part required
in~\cite{ECpaper}) and the main difference is in the analysis. In~\cite{ECpaper}, we completely bypassed the need for using augmenting paths
in the analysis and instead, took advantage of structural properties of matchings in a global manner (by using \emph{Tutte-Berge formula};
see, e.g.,~\cite{lovasz2009matching}). In particular, we showed that repeatedly picking $O(\frac{\log{(1/\eps p)}}{\eps p})$ maximum
matchings (as described before) suffices to ensure that, among the chosen edges, a matching of size (essentially) equal to the size of the
last chosen matching would be realized (with high probability). Having this, one can show that running the aforementioned algorithm even for
$R:= O(\frac{\log{(1/\eps p)}}{\eps p})$ suffices to obtain a $(\frac{1}{2}-\eps)$-approximation.
  
\emph{Adaptive} algorithms for stochastic matching have also been studied by~\cite{BlumDHPSS15,ECpaper}. In an adaptive algorithm, instead
of a single graph $H$, one is allowed to pick a \emph{small} number of bounded-degree graphs $H_1,\ldots,H_k$ where the choice of each $H_i$
can be made after \emph{probing} the edges in $H_1, H_2, \ldots, H_{i-1}$ to see if they are realized or not.  A $(1-\eps)$-approximation adaptive
algorithm for this problem was first proposed in~\cite{BlumDHPSS15} and further refined in~\cite{ECpaper}.
   
\paragraph{Beating the half approximation.} 
This state-of-the-art highlights the following natural question: 
\begin{quote}
  \emph{Is half-approximation the limit for non-adaptive algorithms or is there a non-adaptive algorithm that achieves approximation
    guarantee of {strictly better than half}? }
\end{quote}

It is worth mentioning that in many variations, obtaining half approximation for the maximum matching problem is typically a relatively easy task (usually via a greedy approach), while beating half approximation turns out to be a
difficult task. Some notable examples include, randomized greedy matching~\cite{DyerF91,AronsonDFS95,PoloczekS12,ChanCWZ14}, online stochastic matching~\cite{KVV90,MehtaP12,MehtaWZ15}, and
semi-streaming matching~\cite{FKMSZ05,KonradMM12}.

\subsection{Our Contributions} 
We resolve the aforementioned question of obtaining an algorithm for stochastic matching with an approximation guarantee of \emph{strictly}
better than half. Formally,

\begin{theorem}\label{thm:main}
  There exists an algorithm that given any graph $G(V,E)$ and any parameter $p > 0$, computes a subgraph $H(V,Q)$ of $G$ with a maximum
  degree of $\bigO{\frac{\log{(1/p)}}{p}}$ such that the ratio of the expected maximum matching size of a realization of $H$ to a
  realization of $G$ is at least:
	\begin{enumerate} [(i)]
	\item \label{part:main1} $0.52$ when $p \leq p_0$ for an \emph{absolute} constant $\constp > 0$.   
	\item \label{part:main2} $0.5 + \delta_0$ for any $0 < p < 1$, where $\delta_0 > 0$ is an \emph{absolute} constant.
	\end{enumerate}
\end{theorem}


Our result in Theorem~\ref{thm:main} makes progress towards an open problem posed by Blum~\etal~\cite{BlumDHPSS15} regarding the possibility
of having a non-adaptive $(1-\eps)$-approximation algorithm for stochastic matching. We further remark that the assumption in Part~(\ref{part:main1}) of Theorem~\ref{thm:main} is standard
in the stochastic matching literature and is referred as the case of \emph{vanishing probabilities}, see, e.g.~\cite{MehtaP12,MehtaWZ15}. 

It is worth mentioning the max-degree on $H$ achieved in Theorem~\ref{thm:main} is essentially the best possible (up to an
$O(\log{{1 \over p}})$ factor) for any \emph{constant factor} approximation algorithm: suppose $G$ is a complete graph; in this case the
expected matching size in $G$ is $ n-o(n)$ by standard results on random graphs (see, e.g.,~\cite{Bollobas2001}, Chapter~7); however, if
max-degree of $H$ is $o(\frac{1}{p})$, then the expected number of realized edges in $H$ is $o(n)$, implying that the expected matching size
in $H$ is $o(n)$.
 
Our approach to proving Theorem~\ref{thm:main} can be divided into two parts. In the first part, we prove a structural result showing that
if a realization of $G$ has expected maximum matching size $\opt$, then $G$ itself should contain essentially $\frac{1}{p}$ edge-disjoint
matchings of size $\opt$ each.  This result, established through a characterization of $b$-matching size in general graphs (see
Section~\ref{sec:prelim}), sheds more light into the structure of a graph in terms of its expected maximum matching size, which may be of
independent interest.

In the second part, we combine the aforementioned structural result with the $(\frac{1}{2}-\eps)$-approximation algorithm of~\cite{ECpaper}
to obtain a matching of size strictly larger than $\opt/2$. In order to do this, we first find a collection of $\frac{1}{p}$ edge-disjoint
matchings of size at least $\opt$, remove them from the graph, and then run the algorithm of~\cite{ECpaper} on the remaining edges. We show
that the edges in this collection of edge-disjoint matchings must form many \emph{length-three augmenting paths} of the matching computed by
the algorithm of~\cite{ECpaper}, hence leading to a matching of size strictly larger than $\opt/2$. The analysis is separated into two steps:
we first formulate the increment in the matching size (through these augmenting paths) via a (non-linear) minimization program, and then
analyze the optimal solution of this minimization program and hence lower bound the increment in the matching size obtained from the
augmenting paths.

\paragraph{Other related work.} Multiple variants of stochastic matching have been considered in the literature. Blum~\etal~\cite{BlumGPS13}
studied a similar setting where one can only probe \emph{two} edges incident on any vertex and the goal is to find the optimal set of edges
to query.  Another well studied setting is the \emph{query-commit} model, whereby an algorithm probes one edge at a time and if an edge $e$
is probed and realized, then the algorithm must take $e$ as part of the matching it
outputs~\cite{CostelloTT12,Adamczyk11,ChenIKMR09,BansalGLMNR12,GuptaN13}. We refer the reader to~\cite{BlumDHPSS15} for a detailed
description of the related work.

\paragraph{Organization.} The rest of the paper is organized as follows. We start by providing a high level overview of our algorithm in
Section~\ref{sec:overview}. Next, in Section~\ref{sec:prelim}, we introduce the notation and preliminaries needed for the rest of the
paper. We prove our main structural result, i.e., $b$-matching lemma in Section~\ref{sec:b-matching}. Our main algorithm and its analysis, i.e., the proof 
of Part~(\ref{part:main1}) of Theorem~\ref{thm:main} are provided in Section~\ref{sec:main-alg}. Proof of Part~(\ref{part:main2}) of Theorem~\ref{thm:main}, i.e., an algorithm that works
for the large-probability case appears in Section~\ref{app:large-p}. We conclude the paper in Section~\ref{sec:conc}.

\renewcommand{\event}{\mathcal{E}}
\section{Technical Overview}\label{sec:overview}

In this section, we give a more detailed overview of the main ideas used in our algorithm for stochastic matching. For clarity of
exposition, throughout this section, we assume $p$ is a sufficiently small constant (corresponding to Part~(\ref{part:main1}) of
Theorem~\ref{thm:main}) and the expected maximum matching size in $G$ (i.e., $\opt$) is $n-o(n)$, or in other words, a realization of $G$,
$G_p$, has a near perfect matching in expectation.

Our starting point is the following observation: In order for $G_p$ to have a (near) perfect matching in expectation, the input graph $G$
must have many (roughly $1/p$) edge-disjoint (near) perfect matchings.  To gain some intuition why this is true, suppose for the moment that
the input graph is a bipartite graph $G(L,R,E)$. Then, by \emph{Hall's Marriage Theorem}, we know that in order for $G_p$ to have a matching
of size $n - o(n)$, for any two subsets $X \subseteq L$ and $Y \subseteq R$, with $\card{X} - \card{Y} \geq o(n)$, at least one edge from
$X$ to $\bar{Y}$ should realize in $G_p$. However, this requirement implies that in $G$, there should be $1/p$ edges from $X$ to $\bar{Y}$
so that at least one of these edges appears in $G_p$. One can then show that a bipartite graph $G$ with such a structure has $1/p$
edge-disjoint matchings of size at least $n-o(n)$.

In general, we need to handle graphs that are not necessarily bipartite. In order to adapt the previous strategy, we slightly relax our
requirement of having $1/p$ edge-disjoint matchings to having one (simple) $b$-matching\footnote{Recall that a (simple) $b$-matching is
  simply a graph with degree of each vertex bounded by $b$. See Section~\ref{sec:prelim} for more details.} of size $nb$ for the parameter
$b = \frac{1}{p}$. We show that,
\begin{itemize}
	\item[] \textbf{b-Matching Lemma.} Any graph $G$ where $G_p$ has a matching of size $n-o(n)$ in expectation, has a $\frac{1}{p}$-matching of size (essentially) $\frac{n}{p}$. 
\end{itemize}


Next, we combine the fact that a large $\frac{1}{p}$-matching, denoted by $B$, always exists in $G$, with the
$(\frac{1}{2}-\eps)$-approximation algorithm of~\cite{ECpaper} to obtain a strictly better than $\frac{1}{2}$-approximation algorithm.

To continue, we briefly describe the algorithm of~\cite{ECpaper}, which we refer to as \basicalg. \basicalg works by picking a maximum
matching $M_i$ in $G$ and removing the edges of $M_i$ for $R:= \Theta\paren{\frac{\log{(1/p)}}{p}}$ times\footnote{We remark
  that this algorithm has an extra \emph{sparsification} step which is needed to handle the case where $\opt = o(n)$. However, since in this
  section we assume $\opt = n-o(n)$, this extra step is not required.}. This collection of matchings, denoted by $E_{MC}$, is referred to as
a \emph{matching cover} of the original graph $G$. The main property of this matching cover, proved in~\cite{ECpaper}, is that the set of
realized edges in $E_{MC}$ has a matching of size (essentially) $\card{M_R}$; note that $M_R$ is the smallest size matchings among the
matchings in $E_{MC}$.

 
We are now ready to define our main algorithm: Pick a maximum $\frac{1}{p}$-matching $B$ from $G$; run \basicalg over the edges
$E\setminus B$ and obtain a matching cover $E_{MC}$; return $H(V,B \cup E_{MC})$.  If $\card{M_R} < ({1 \over 2} - \delta_0)n$, using the
fact that $E_{MC}$ is obtained by repeatedly picking maximum matchings, one can show that any matching $M$ of size $n-o(n)$ in $G$ has more
than $(\frac{1}{2} + \delta_0)n - o(n)$ edges in $B \cup E_{MC}$. This also implies that the expected matching size in $H$ is at least
$(\frac{1}{2} + \delta_0)n - o(n)$.  The more difficult case, which is where we concentrate bulk of our technical effort, is when
$\card{M_R} \geq (\frac{1}{2}-\delta_0)n$. For simplicity, assume $\card{M_R} = n/2$ from here on.

As stated above, if $|M_R| = n/2$, then in almost every realization of the edges in $E_{MC}$, there exists a matching $M$ of size at least
$n/2$.  Our strategy is to \emph{augment} the matching $M$ using the (realized) edges in $B$, so that the matching size becomes
$({1 \over 2} + \delta_0) n$.  It is important to note that the set of edges in $E_{MC}$ and $B$ are disjoint, and hence whether edges in
$E_{MC}$ and $B$ are realized are independent of each other.

Let $U$ be the set of vertices matched by $M$. There are two cases here to consider:
\begin{itemize}
\item \textbf{Case 1.} Nearly all edges in $B$ are incident on vertices in $U$.
\item \textbf{Case 2.} An $\eps$-fraction of edges in $B$ are not incident on $U$ (for some constant $\eps > 0$).
\end{itemize}
 
The second case is relatively easy to handle: we show that a realization of a $\frac{1}{p}$-matching with $N/p$ edges has a matching of size
at least $N/3$ in expectation. This implies that $B_p$ has a matching $M'$ of size $\eps \cdot \frac{n}{3} = \Theta(\eps) \cdot n$ which is
not incident on $U$. Consequently, $B \cup E_{MC}$ has a matching of size $\frac{n}{2} + \Theta(\eps) \cdot n$ in expectation.  The more
challenging task is to tackle the first case. To convey the main idea, we make a series of simplifying assumptions here: $(i)$ all edges in
$B$ are incident on $U$, $(ii)$ each edge in $B$ is incident on exactly one vertex in $U$, and $(iii)$ every vertex in $U$ is incident on
exactly $\frac{1}{p}$ edges of $B$.
 
Our goal is to identify a large collection of length-three augmenting paths for the matching $M$ using the edges of $B$.  To achieve this,
we consider the event that an edge $(u,v)$ in $M$ has a length-three augmenting path $a-u-v-b$ where $u$ (resp. $v$) is the only neighbor of
$a$ (resp. $b$).  We say such an edge $(u,v)$ is \emph{successful}. Since the length-three augmenting that certifies successful edges are
vertex-disjoint by definition, they can all (simultaneously) augment $M$.  Consequently, it suffices to lower bound the expected number of
successful edges, or, equivalently, to lower bound the prob. that each edge is successful.



Let us further assume for the moment that $G$ is a bipartite graph. In this case, $u$ and $v$ do not share a common neighbor and we can
consider the neighborhood of $u$ and $v$ separately. The prob. that $u$ has a neighbor $w$ where $u$ is the only neighbor of $w$ (we say
$u$ is successful in this case) is not difficult to bound: enumerate all $1/p$ neighbors $w$ of $u$ and account for the the prob. that the
edge $(u,w)$ is realized and the prob. that no other edge incident on $w$ is realized.  A similar argument can be made for $v$.  Now, the
prob. that $(u,v)$ is successful is simply the product of the prob. that $u$ is successful and the prob. that $v$ is successful.



However, in general (non-bipartite) graphs, $u$ and $v$ might have common neighbors which results in prob. of $u$ being successful \emph{not independent} of 
prob. of $v$ being successful. Handling this case requires a more careful argument and analysis. Moreover, recall that in the above
discussion, we made rather strong simplifying assumptions about how the edges in $B$ are distributed across the vertices of $U$.  In order
to further remove these assumptions, in the actual analysis, we cast the probability of each edge $(u,v)$ being successful as a function of
the degrees of the vertices $u$ and $v$, and formulate a (non-linear) minimization program to capture the minimum number of possible
successful edges. Finally, we analyze the optimal solution of this minimization program, which allows us to achieve the target lower bound
on the expected increment in the matching size.


\section{Preliminaries}\label{sec:prelim}

\paragraph{Notation.} 
For a graph $G(V,E)$, $n$ denotes the number of vertices in $G$. For any $U \subseteq V$, we use $G[U]$ to denote the subgraph of $G$
induced only on vertices in $U$, and use $E[U]$ to denote the set of edges in $G[U]$, i.e., the set of edges with both end points in
$U$. For any two subsets $U,W$ of $V$, we further use $E[U,W]$ to denote the set of edges with one end point in $U$ and another in $W$. For
any $X \subseteq E$, we use $V(X)$ to denote the set of vertices incident on $X$.  Finally, we use $\mu(E)$ to denote the maximum matching
size among a set of edge $E$.

When sampling from a set of edges $X$ (resp. a graph $H$) where each edge in $X$ (resp. $H$) is sampled w.p. $p$, we use $X_p$ (resp. $H_p$)
to denote the random variable for the set of sampled edges. We use $\OPT(G)$ (or shortly $\OPT$ if the graph $G$ is clear from the context)
to denote the \emph{expected} maximum matching size of a realization of $G$ (i.e., $G_p(V,E_p)$)\footnote{We assume $\opt = \omega(1)$ to
  obtain the desired concentration bounds (for example in Lemma~\ref{lem:basic-alg}). }. For any algorithm for the stochastic
matching problem, we use $\ALG$ to denote the expected matching size in a realization of $H$, where $H$ is the subgraph computed by the
algorithm.




\paragraph{b-matchings.} For any graph $G(V,E)$ and any integer $b \geq 1$, a subset $M \subseteq E$ is called a \emph{simple $b$-matching}, iff the number of edges $M$ that are incident on each
vertex is at most $b$. Throughout, we drop the word `simple', and refer to $M$ as a $b$-matching. 

We use the following characterization of the maximum $b$-matching size in general graphs (see~\cite{COBook}, Volume A, Chapter 33). 
\begin{theorem}\label{thm:b-matching-char}
	Let $G(V,E)$ be a graph and $b \geq 1$ be any integer. The maximum size of a $b$-matching is equal to the minimum value of 
	\begin{align*}
		b \cdot \card{U} + \card{E[W]} + \sum_{K} \floor{\frac{1}{2}\Paren{b\cdot \card{K} + \card{E[K,W]}}}
	\end{align*}
	taken over all disjoint subsets $U,W$ of $V$, where $K$ ranges over all connected components in the graph $G[V - U - W]$. 
\end{theorem}

\paragraph{Useful inequalities.} We also use the following simple inequalities. The Proofs are provided in Appendix~\ref{app:prelim} for completeness.

\begin{proposition}\label{prop:upper-exp}
  Let $f(x) := {1 - e^{-x} \over x}$. Then, for any $c$, and any $x \in [0, c]$, $e^{-x} \le 1 - f(c) \cdot x$.
\end{proposition}

\begin{proposition}\label{prop:upper-exp-2}
  For any $x \in (0, 0.43]$, $(1-x)^{{1\over x}} \ge {1 - x \over e}$. \footnote{This inequality actually holds for any $x \in [0,1]$. However, as we only need the range $(0,0.43]$ 
  in our proofs and this allows us to provide a simpler proof, we only consider this range.}
\end{proposition}

\subsection{\basicalg Algorithm}\label{sec:tools}

We use the $(0.5-\eps)$-approximation algorithm of~\cite{ECpaper} (Algorithm~3) as a sub-routine. For simplicity, throughout the paper, we
refer to this algorithm as \basicalg. In the following lemma, we summarize the properties of \basicalg that we use in this paper. The proof
of this lemma immediately follows from Lemma~3.9 and Lemma~5.2 in~\cite{ECpaper}.

\begin{lemma}[\cite{ECpaper}]\label{lem:basic-alg}
  For any graph $G(V,E)$, and any input parameter $\eps > 0$, $\basicalg(G,\eps)$ outputs a collection of $R$ matchings $M_1, M_2, \ldots, M_R$
  (denote $E_{MC} = M_1 \cup M_2 \cup \ldots\cup M_R$), such that, w.p. $1-o(1)$:
	\begin{enumerate}
		\item The size of a maximum matching among realized edges in $E_{MC}$ is at least $\paren{1-\eps} \card{M_R}$.
		\item $\card{M_1} \geq \ldots \geq \card{M_R} \geq (1-\eps)\cdot\mu\paren{E \setminus E_{MC}}$.
		\item $R = \Theta(\frac{\log{1/(\eps p)}}{\eps p})$.
	\end{enumerate}
\end{lemma}

We can also prove the following simple claim based on the second property of the \basicalg in Lemma~\ref{lem:basic-alg}.  Roughly speaking, this claim states that if the \basicalg is not able to 
extract any further large matching (of size essentially $\opt/2$) from $G$, then the set of extracted edges already provides a matching of size $\opt/2$ in any realization. 
A similar result is proven in~\cite{ECpaper} (see Lemma~5.3); however, since Claim~\ref{clm:small-L} does not follow directly from the results in~\cite{ECpaper}, we
provide a self-contained proof of this claim here. 

\begin{claim}\label{clm:small-L}
	Fix $0 < \eps < \delta < 1$. Let $G(V,E)$ be a graph, $X$ be any arbitrary subset of $E$, and $(M_1,\ldots,M_R) = \basicalg(G(V,E\setminus X),\eps)$. 
	Define $E_{MC} = M_1 \cup \ldots \cup M_R$. 	If $\card{M_R} \leq \paren{\frac{1}{2}-\delta} \OPT$, then the expected maximum matching size in a realization
	of $G(V, X \cup E_{MC})$ is at least $\paren{\frac{1}{2} + \delta - \eps} \OPT$. 
\end{claim}
\begin{proof}
For each realization of $G_p$, we fix one maximum matching. Now the expected matching size in $G_p$
  can be written as
  \begin{align*}
    {\OPT} = \sum_M\prob{\text{$M$ is the fixed maximum matching in $G_p$}} \cdot \card{M}
  \end{align*}
  By property~(2) of \basicalg in Lemma~\ref{lem:basic-alg}, the maximum matching size in the graph $G(V,E \setminus (X \cup E_{MC}))$ is at most $(1+\eps)\card{M_R}$. 
  Therefore, for any matching $M$, at most $(1+\eps)\card{M_R}$ edges of $M$ is in $E \setminus (X \cup E_{MC})$, and hence at least $\card{M} - (1+\eps)L$ edges
  of $M$ is in $X \cup E_{MC}$. This implies that if all edges in $M$ are
  realized, a matching of size at least $\card{M} - (1+\eps)\card{M_R}$ is realized in $Q$. Let $\ALG$ be the expected maximum matching size in $G(V,X \cup E_{MC})$; we have,  
  \begin{align*}
  \ALG &\ge \sum_M\prob{\text{$M$ is the fixed maximum matching in $G_p$}} (\card{M} - (1+\eps)\card{M_R})\\
                  &={\OPT} - (1+\eps)\card{M_R}
  \end{align*}
  Since $\card{M_R} \le (1/2 - \delta){\OPT}$, we have, 
  \[ \ALG \ge {\OPT} - (1+\eps)\cdot(1/2 - \delta) \cdot {\OPT} \geq (1/2 + \delta - \eps) \cdot {\OPT}\]
  which concludes the proof.
\end{proof}

\newcommand{\xstar}{x^{\star}}
\section{$b$-Matching Lemma}\label{sec:b-matching}

Here, we develop one of the main ingredients of our algorithm, namely, any input graph $G$ contains a $b$-matching of size almost
$b \cdot \OPT(G)$ for $b = 1/p$.  Intuitively, if the expected matching size in $G$ is $\OPT$, then since only $p$ fraction of edges are
realized in expectation, one may hope to find up to $1/p$ edge-disjoint matchings of size \OPT in $G$.  The following lemma formalizes this
intuition by using $b$-matchings (for $b = 1/p$) instead of a collection of edge-disjoint matchings.

\begin{lemma}[$b$-matching lemma]\label{lem:b-matching}
	Let $b = \floor{\frac{1}{p}}$; any graph $G(V,E)$ has a $b$-matching of size at least $(b-1)\cdot\opt (G)$.  
\end{lemma}
\begin{proof}
	Suppose by contradiction that the maximum $b$-matching $B$ in $G$ is of size less than $(b-1)\cdot\opt$. Consequently, by Theorem~\ref{thm:b-matching-char}, there exist disjoint subsets $U,W$ of $V$ such that,
	\begin{align}
		b \cdot \card{U} + \card{E[W]} + \sum_{K} \floor{\frac{1}{2}\Paren{b\cdot \card{K} + \card{E[K,W]}}} < (b-1)\cdot\opt \label{eq:b-matching-1}
	\end{align}
	where $K$ ranges over all connected components in the graph $G[V - U - W]$. Let $c$ be the number of connected components in $G[V-U-W]$. We first note that $c < 2\opt$; otherwise,
	\begin{align*}
		\sum_{K} \floor{\frac{1}{2}\Paren{b\cdot \card{K} + \card{E[K,W]}}} &\geq  \sum_{K} \floor{\frac{b}{2}} \geq c \cdot \paren{\frac{b-1}{2}} \geq (b-1)\cdot\opt 
	\end{align*}
	and hence the LHS in Eq~(\ref{eq:b-matching-1}) would be more than $(b-1) \cdot \opt$, i.e., the RHS; a contradiction. 
	
	Additionally, we have $\card{U} + \card{W} + \sum_{K} \card{K} = n$. Hence, by multiplying each side in Eq~(\ref{eq:b-matching-1}) by $2$ and plugging in this bound, we have,
	\begin{align*}
		2b\cdot\opt-2\opt &>  nb - b\card{W} + b\card{U} + 2\card{E[W]} + \sum_{K} \paren{\card{E[K,W]} - 1} \\
		&\geq   nb - b\card{W} + b\card{U} + 2\card{E[W]} + \sum_{K} \card{E[K,W]} - 2\opt 
	\end{align*}
	Let $T := V \setminus (U \cup W)$, i.e., the set of vertices in connected components $K$. Using this notation, we can write the above equation simply as, 
	\begin{align}
		b\cdot\card{W} - b\cdot\card{U} - 2\card{E[W]} -  \card{E[T,W]} > b \cdot \paren{n-2\opt} \label{eq:b-matching-2}
	\end{align}
	Now consider the partition $T,U,W$ in a \emph{realized} graph $G(V,E_p)$. Let $E_p[W]$ and $E_p[T,W]$ denote, respectively, the set of edges in $E[W]$ and $E[T,W]$ after sampling the edges 
	w.p. $p$. For any matching $M$ in $G_p$, define $x(M)$ to be the number of \emph{unmatched} vertices (by $M$) in $W$. Finally, define $\xstar := \min_{M} x(M)$, where the minimum is taken over all 
	matchings in $G_p$. Clearly, $\xstar$ is a random variable depending on the choice of edges in $G_p$. We have the following simple claim. 
	\begin{claim}\label{clm:x-size}
		For any realization $G_p$, $\xstar \geq \card{W}-\card{U} - 2\card{E_p[W]} - \card{E_p[T,W]}$. 
	\end{claim}
	\begin{proof}
		Consider the set of vertices in $W$. At most $\card{U}$ vertices of $W$ can be matched to vertices in $U$. Additionally, any edge in $E_p[W]$ can further reduce the number
		of unmatched vertices in $W$ by at most $2$. Finally, any edge in $E_p[T,W]$ can reduce the number of remaining unmatched vertices in $W$ by at most $1$. 
	\end{proof}
	Using the fact that $\Ex{\xstar} \leq n-2\opt$, we have,  
	\begin{align*}
		b\cdot (n-2\opt) &\geq b\cdot\Ex{\xstar} \\
		&\geq b\cdot\card{W} - b\cdot\card{U} - b\cdot\EX{2\card{E_p[W]} + \card{E_p[T,W]}} \tag{by Claim~\ref{clm:x-size}}\\
		&= b\cdot\card{W} - b\cdot\card{U} - pb \cdot \paren{2\card{E[W]} + \card{E[T,W]}} \\
		&\geq b\cdot\card{W} - b\cdot\card{U} - {2\card{E[W]} - \card{E[T,W]}} \tag{since $pb = p \floor{\frac{1}{p}} \leq 1$} \\
		&> b \cdot (n-2\opt) \tag{by Eq~(\ref{eq:b-matching-2})}
	\end{align*}
	a contradiction. 
\end{proof}

We further prove that the bound established in Lemma~\ref{lem:b-matching} is essentially tight (see Appendix~\ref{app:b-limitation}). 
\begin{claim}\label{clm:b-match-opt} 
  For any constant $0 < p < 1$, there exist bipartite graphs $G$ where $G_p$ has a matching of size $n-o(n)$ in expectation, but for any
  $b \geq {2 \over p}$, there is no $b$-matching in $G$ with (at least) $b\cdot 0.99n$ edges; here $n$ is the number of
  vertices on each side of $G$.
\end{claim}

Finally, we establish the following auxiliary lemma. 

\begin{lemma}\label{lem:outside-edges}
  Let $B$ be a $\floor{\frac{1}{p}}$-matching with $\paren{\floor{\frac{1}{p}}\cdot N}$ edges; then, $\EX{\mu(B_p)} \geq (1-3p) \cdot {\frac{N}{3}}$. 
\end{lemma}
\begin{proof}
	We first partition the edges of $B$ into a collection of matchings. Since the degree of each vertex in $G(V,B)$ is at most $\floor{\frac{1}{p}}$, by Vizing's Theorem~\cite{Vizing64}, we can color
	the edges in $G(V,B)$ with $\floor{\frac{1}{p}}+1$ colors such that no two edges with the same color are incident on a vertex.
	This ensures that $B$ can be decomposed into $R = \floor{\frac{1}{p}}+1$ matchings
	$M_1,\ldots,M_R$. 

	Next, we define the following process. Define $M^{(0)} = \emptyset$; for $i = 1$ to $R$ rounds, let $M^{(i)}$ be a maximal matching
        obtained by adding to $M^{(i-1)}$ the set of realized edges in $M_i$ that are not incident on vertices in $M^{(i-1)}$. Define
        $M := M^{(R)}$. 
        
        We argue that $\Ex{\card{M}} \geq (1-3p) \cdot {\frac{N}{3}}$. To do this, we need the following notation. Define $Y_i$ as a random variable denoting 
	the \emph{set of edges} in $M_i$ that are \emph{not} incident on any vertex of matching $M^{(i-1)}$. Note that $Y_i$ depends only on the realization of edges in $M_1,\ldots,M_{i-1}$ and is \emph{independent} of 
	the realization of $M_i$. Moreover, define $X_i$ as a random variable indicating the \emph{number of edges} in (a realization of )$Y_i$ that are added to $M^{(i-1)}$ (after updating by edges in $M_i$). We first have, 
	\begin{align*}
		\card{Y_i} \geq \card{M_i} - 2 \card{M^{(i-1)}}
	\end{align*}
	since any edge in $M^{(i-1)}$ can be incident on at most two vertices of $M_i$. Moreover, conditioned on any valuation
	for $Y_i$, we have $\Ex{X_i} = p \cdot \card{Y_i}$ since each edge in $M_i$ is realized w.p. $p$, independent 
	of the choice of $Y_i$. Consequently, 
	\begin{align*}
		\Ex{X_i} = p \cdot \Ex{Y_i} \geq p \cdot \paren{\card{M_i} - 2\EX{\card{M^{(i-1)}}}}
	\end{align*}
	We again stress that the expectation for $X_i$ is taken over the choice of edges in $M_i$, while the expectation for $Y_i$ (and $M^{(i-1)}$) is taken over the choice of edges in $M_1,\ldots,M_{i-1}$. 
	We now have, 
	\begin{align*}
		\Ex{\card{M}} &= \sum_{i=1}^R \Ex{X_i} \geq \sum_{i=1}^{R} p \cdot \paren{\card{M_i} - 2\EX{\card{M^{(i-1)}}}} \\
		&\geq p \cdot \paren{\sum_{i=1}^{R} \card{M_i} - 2 \sum_{i=1}^{R} \Ex{\card{M}}} \tag{$\Ex{\card{M}} \geq \EX{\card{M^{(i-1)}}}$} \\
		&\geq p \cdot \paren{\floor{\frac{1}{p}} \cdot N - 2 \paren{\floor{\frac{1}{p}}+1} \cdot \Ex{\card{M}}} \tag{$R = \floor{\frac{1}{p}}+1$}
	\end{align*} 
	This implies that 
	\begin{align*}
		\Ex{\card{M}} \geq (1-3p) \cdot \frac{N}{3}
	\end{align*}
	which concludes the proof.
\end{proof}

\section{Main Algorithm and Analysis}\label{sec:main-alg}

We provide our main algorithm for the stochastic matching problem (when $p$ is sufficiently small) in this section and prove
Part~(\ref{part:main1}) of Theorem~\ref{thm:main}.  We assume throughout this section that the edge realization probability
$p \le p_0$ for some sufficiently small constant $p_0$. In this case, $\floor{{1 \over p}} - 1 \geq (1 - O(p_0)) \cdot {1 \over p}$ and we
use this inequality frequently in the proof. Indeed, throughout this section, one should view $p_0$ as a negligible constant and hence the
term $(1 - O(p_0))$ can essentially be ignored.

Let $\delta_0 = 0.02$, and $\eps_0 = 0.02001$.  Our algorithm is stated as Algorithm~\ref{alg:small-p} below:


\smallskip
\begin{algorithm}[H]
\textnormal
\SetAlgoNoLine
\KwIn{A graph $G(V,E)$ and an edge realization probability $p \leq p_0$.}
\KwOut{A subgraph $H(V,Q)$ of $G(V,E)$.}
\begin{enumerate}
\item Let $B$ be a maximum $\bMatch$-matching in $G$. 
\item Let $(M_1, M_2, \ldots, M_R) := \basicalg(G(V,E\setminus B),\eps_1)$ for $\eps_1 = (\eps_0 - \delta_0)/2$, and $E_{MC} = M_1 \cup \ldots \cup M_R$.  
\item Return $H(V,Q)$ where $Q:= B \cup E_{MC}$. 
\end{enumerate}
\caption{A $0.52$-Approximation Algorithm for Stochastic Matching}
\label{alg:small-p}
\end{algorithm}

Each vertex in $H$ has degree  $\bigO{{\log(1/p) \over p}}$ -- this follows immediately from Lemma~\ref{lem:basic-alg}. In what
follows, we prove that $H$ has a matching of size at least $(0.5+\delta_0) \cdot \opt = 0.52 \cdot \opt $ in expectation, which will complete the
proof of Part~(\ref{part:main1}) of Theorem~\ref{thm:main}.
 
First notice that if $\card{M_R} < ({1 \over 2} - {\eps_0 + \delta_0 \over 2})\OPT$ where $M_R$ is the smallest matching in the matching
cover $E_{MC}$ found by Algorithm~\ref{alg:small-p}, then by Claim~\ref{clm:small-L}, the expected matching size in $Q$ is at least
$({1 \over 2} + {\eps_0 + \delta_0 \over 2} - {\eps_0 - \delta_0 \over 2})\OPT = ({1 \over 2} + \delta_0) \cdot \OPT$.  Therefore, from now
on we focus on the case that $\card{M_R} \ge ({1 \over 2} - {\eps_0 + \delta_0 \over 2})\OPT$.

In this case, by Lemma~\ref{lem:basic-alg}, w.p. $1-o(1)$, there exists a matching $M$ among the realized edges in $E_{MC}$ with size at
least
\begin{align*}
  \paren{1 - {\eps_0 - \delta_0 \over 2}} \paren{{1 \over 2} - {\eps_0 + \delta_0 \over 2}} \OPT &\ge  \paren{{1 \over 2} - {\eps_0 + \delta_0 \over 2} - {\eps_0  - \delta_0 \over 4}} \OPT \\
  &=  \paren{{1 \over 2} - {3\eps_0 \over 4} - {\delta_0 \over 4}} \OPT \ge \paren{{1 \over 2} - \eps_0}\OPT
\end{align*}
In the following, we assume this event happens\footnote{This assumption can be removed while losing a negligible factor of $o(1)$ in the
  size of final matching.} and prove that the set of edges realized in the $\bMatch$-matching $B$ can be used to augment the matching $M$
to create a matching of size $({1 \over 2} + \delta_0) \cdot \OPT$ in expectation.  To simplify the analysis, we assume w.l.o.g. that
$\card{M} = ({1\over2} - \eps_0)\OPT$ (i.e., we only keep $({1\over2} - \eps_0)\OPT$ edges of $M$ and remove any additional edges if there
is any). By the $b$-matching lemma (Lemma~\ref{lem:b-matching}),
$\card{B} \ge \paren{\bMatch-1} \OPT \ge (1 - O(p_0)) \cdot {\OPT \over p}$, and hence, to prove Part~(1) of Theorem~\ref{thm:main}, it
suffices to prove the following statement.

\begin{lemma}\label{lem:b-matching-augmentation-raw}
  Let $M$ be a matching of size $\paren{{1 \over 2} - \eps_0} \OPT$, and $B$ be a $\bMatch$-matching of size at least $(1 - O(p_0)) \cdot {\OPT \over p}$; then 
  the expected maximum matching size in $M \cup \BoutM$ is at least $({1 \over 2} + \delta_0) \OPT$.
\end{lemma}

\begin{proof}

Let $\BinM$ be the set of edges in $B$ that are incident on the vertices in the matching $M$, and let $\BoutM = B \setminus \BinM$.  Let
$\outMs$ be the random variable denoting the maximum matching size of a realization of $\BoutM$. By Lemma~\ref{lem:outside-edges},
\begin{equation}
  \label{eq:bout}
  \expect{\outMs} \ge {\card{\BoutM} \over \bMatch} \cdot {1 - 3p \over 3} \ge {p \card{\BoutM}} \cdot {1 - 3p \over 3}
\end{equation}
Therefore, if $\card{\BoutM} \ge 6\eps_0 \cdot {\OPT \over p}$, then 
\begin{align*}
  \expect{\outMs} \ge 6\eps_0 \cdot {\OPT} \cdot {1 - 3p \over 3} =  2\eps_0  (1 - 3p)  \cdot \OPT \ge (\eps_0 + \delta_0) \cdot \OPT \tag{assuming $p_0 \le { \eps_0 - \delta_0 \over 6\eps_0}$}
\end{align*}
and since no edge in $\BoutM$ is incident on the vertices in $M$, the expected matching size in $M \cup B_p$ is at least
\begin{align*}
  \paren{{1 \over 2} -\eps_0} \OPT + (\eps_0 + \delta_0) \OPT = \paren{{1 \over 2} + \delta_0} \OPT  
\end{align*}
as asserted by Lemma~\ref{lem:b-matching-augmentation-raw}. In the following, we assume $\card{\BoutM} \le 6\eps_0 \cdot {\OPT \over p}$.
Furthermore, we fix a realization of $\BoutM$ and fix a maximum matching $M'$ in the realization of $\BoutM$ (whose size is $\outMs$ by
definition). In other words, we will lower bound the expected maximum matching size in $M \cup B_p$ conditioned on \emph{any} realization of
$\BoutM$. The lower bound we obtain would be a linear function of $\outMs$, and by linearity of expectation, we can simply replace $\outMs$
with $\expect{\outMs}$, use Eq~(\ref{eq:bout}) to lower bound $\expect{\outMs}$, and obtain the desired lower bound of
$(1/2 + \delta_0)\OPT$ on the expected maximum matching size.

Denote by $\newM$ the matching $M \cup M'$ (since the matchings $M$ and $M'$ are vertex-disjoint, $M \cup M'$ is indeed a valid matching of
size $\card{M} + \outMs$). We can focus the realizations of $\BoutM$ where $\outMs \le 2\eps_0 \cdot \OPT$ since otherwise the matching $\newM$
already have size $(1/2 + \eps_0) \cdot \OPT > (1/2 + \delta_0) \cdot \OPT$. Therefore, we have $\outMs \le 2\eps_0 \cdot \OPT = O(\OPT)$
and $\card{\newM} = O(\OPT)$, which will be useful in simplifying the presentation.



Now consider the edges in $\BinM$. We further denote by $C$ the set of edges in $\BinM$ that are incident on \emph{exactly} one vertex in
$\newM$. In the following, we first show that $\card{C}$ must be large (Claim~\ref{clm:C-size}) and then show that many edges in $C$ can be
used to augment the matching $\newM$, which leads to an increment on the matching size as a function of $\card{C}$
(Lemma~\ref{lem:1-neighbor} and Lemma~\ref{lem:allocation}). Combining these two statements completes the proof of
Lemma~\ref{lem:b-matching-augmentation-raw}.


\begin{claim}\label{clm:C-size}
  $\card{C} \ge 2\card{\BinM} - {2\card{\newM} \over p}$.
\end{claim}

\begin{proof}
  Let $x$ denote the number of edges in $\BinM$ that have degree $2$ to $V(\newM)$ (i.e., are incident on two vertices in $\newM$).  By
  definition, every edge in $\BinM$ is incident on $M$, and hence every edge in $\BinM$ is also incident on $\newM ( = M \cup M')$.
  Consequently, there are $\card{\BinM} - x$ edges in $\BinM$ that have degree $1$ to $V(\newM)$ (i.e. belongs to $C$). Therefore, the total
  degrees of all vertices $V(\newM)$ provided by $\BinM$ is at least:
  $  2 \cdot x +  1 \cdot \paren{\card{\BinM} - x} = x + \card{\BinM}$.

  On the other hand, since $\card{V(\newM)} = 2\card{\newM}$ and $B$ (hence $\BinM$) is a $\bMatch$-matching, the total degree of the
  vertices $V(\newM)$ provided by $\BinM$ is at most ${2\card{\newM} \over p}$. Therefore,
   $ x +\card{\BinM} \le  {2\card{\newM} \over p}$,
  which implies $x \le {2\card{\newM} \over p}  - \card{\BinM}$. Therefore, the number of edges in $\BinM$ incident on exactly one
  vertex in $V(\newM)$ (i.e., $\card{C}$) is at least
  \begin{align*}
    \card{\BinM} - \paren{{2\card{\newM} \over p}  - \card{\BinM}  } = 2\card{\BinM} - {2\card{\newM} \over p} 
  \end{align*}
  completing the proof. 
\end{proof}

The following two lemmas are dedicated to showing that the edges in a realization of $C$, $C_p$, form many vertex-disjoint length-three augmenting
paths for the matching $\newM$ in expectation, which is a lower bound on the expected increment on the matching size.  We first define some
notation. Let $W:= V \setminus V(\newM)$, i.e., $W$ is the set of vertices \emph{not} matched by $\newM$. Denote the edges in $\newM$ by
$\set{(u_i, v_i) \mid i \in [\card{\newM}]}$, and denote by $d(u_i)$ (resp. $d(v_i)$) the number of edges in $C$ incident on $u_i$
(resp. $v_i$). Since, by definition, the edges in $C$ are only incident on one vertex in $V(\newM)$, $d(u_i)$ (resp. $d(v_i)$) is also the
number of edges in $C$ between $u_i$ (resp. $v_i$) and $W$.  In the following, whenever we say ``neighbors'' or ``degrees'', they are only
w.r.t. the edges $C$. Let $f(x)$ be the function defined in Proposition~\ref{prop:upper-exp}.

\begin{lemma}\label{lem:1-neighbor}
  For any edge $(u_i,v_i) \in M$, w.p. at least 
  \[(1 - O(p_0)) {f({1 / e})\cdot p \over e^2} \paren{1 - e^{-p \cdot d(v_i)}} \cdot \max\set{d(u_i) - 1, 0},\]
  there exists a length-three augmenting path $a_i - u_i - v_i - b_i$ in the realization $C_p$ of $C$, such that $a_i, b_i$ have no neighbors
  other than $u_i$ and $v_i$.
\end{lemma}

Note that we can use all edges $(u_i,v_i)$ in $\newM$ with such an augmenting path $a_i-u_i-v_i-b_i$ to (simultaneously) augment $\newM$ since
these augmenting paths are vertex-disjoint ($a_i$ and $b_i$ are only neighbors of $u_i$ and $v_i$). Therefore, the expected number of edges
in $\newM$ that has such an augmenting path is a lower bound on the expected increment on the matching size.

\begin{proof}[Proof of Lemma~\ref{lem:1-neighbor}]
  We consider three disjoint subsets of edges in $C$ one by one: $(i)$ the edges between $v_i$ and $W$, $(ii)$ the edges incident on a
  specific vertex $w$ in $W$ (excluding the edge $(v_i, w)$), and $(iii)$ the edges incident on neighbors of $u_i$ other than $w$ (excluding
  the edges incident on $v_i$).

  First, consider the edges between $v_i$ and $W$. The prob. that none of these $d(v_i)$ edges are realized is at most 
  \begin{align*}
    (1 - p)^{d(v_i)} \le e^{-p \cdot d(v_i)}
  \end{align*}
  Therefore, w.p. at least $1 - e^{-p \cdot d(v_i)}$, at least one edge between $v_i$ and $W$ is realized. We condition on this event and fix any such edge, denoted by $(v_i, b_i)$.

  Second, consider the edges incident on $b_i$ (excluding the edge $(v_i, b_i)$). There are at most $1/p$ such edges, and the prob. that
  none of them is realized is at least
  \begin{align*}
    (1-p)^{1/p} \ge {1 - p \over e} \ge {1 - p_0 \over e} \tag{$ p \leq p_0 \le 0.43$}
  \end{align*}
  where the first inequality is by Proposition~\ref{prop:upper-exp-2}. In the following, we
  further condition on no other edges incident on $w$ is realized.

  Third, consider all neighbors of $u_i$ other than $b_i$ (there are at least $\max\set{d(u_i) - 1, 0}$ such neighbors) and the edges
  incident on these neighbors (excluding the edges incident on $v_i$). For each one of these neighbors $w$ of $u_i$, the prob. that the edge
  $(u_i, w)$ is realized (w.p. $p$) and $w$ does not have any neighbor other than $u_i$ (and possibly $v_i$) (w.p. at least
  ${1 - p_0 \over e}$ by Proposition~\ref{prop:upper-exp-2}) is at least $p \cdot {1 - p_0 \over e}$.  Therefore, the prob. that at least
  one neighbor of $u_i$ satisfies these two properties is at least
  \begin{align*}
    1 -    \paren{1 - p \cdot {1 - p_0 \over e}}^{\max\set{d(u_i) - 1, 0}} \ge & 1 - e^{-  (1 - p_0) \cdot p\cdot \max\set{d(u_i) - 1,0}/e}\\
    \ge &f({1 \over e})\cdot (1 - p_0) \cdot  p \cdot \max\set{d(u_i) - 1, 0}/e
  \end{align*}
  where $f(x) = {1 - e^{-x} \over x}$ and the second inequality is by Proposition~\ref{prop:upper-exp}, using the fact that
  \begin{align*}
  { (1 - p_0) \cdot p \max\set{d(u_i) - 1, 0} \over e} \le {1\over e} \tag{since $d(u_i) \le {1\over p}$}
  \end{align*}
 
  Putting the three steps together, the prob. that there is an augmenting path $a_i-u_i-v_i-b_i$ where $a_i$ and $b_i$ has no neighbors
  other than $u_i$ and $v_i$ is at least
  \begin{align*}
    &\paren{1 -  e^{-p \cdot d(v_i)}} \cdot {1 - p_0  \over e} \cdot {f({1 \over e})\cdot (1 - p_0) \cdot  p \cdot \max\set{d(u_i) - 1, 0} \over e} \\
    =& \paren{1 - O(p_0)} \cdot f({1 \over e}) \cdot {p \over e^2} \paren{1 -  e^{-p \cdot d(v_i)}} \cdot \max\set{d(u_i) - 1, 0}
  \end{align*}
\end{proof}

As we pointed out after the statement of Lemma~\ref{lem:1-neighbor}, we need to lower bound the expected number of edges in $\newM$ that has
such an augmenting path, which, by Lemma~\ref{lem:1-neighbor}, is lower bounded by the function $F$ defined below. For the two vectors
$d_u:= (d(u_1),\ldots,d(u_{\card{\newM}}))$ and $d_v:= (d(v_1),\ldots,d(v_{\card{\newM}}))$,
\begin{align*}
  F(d_u,d_v):= \sum_{i \in [\card{\newM}]} \paren{1 - O(p_0)} \cdot f({1 \over e})\cdot {p \over e^2} \paren{1 - e^{-p \cdot d(v_i)}} \cdot \max\set{d(u_i) - 1, 0} 
\end{align*} 
The goal now is to find the smallest value of $F(d_u,d_v)$, with the constraint on the vectors $d_u$ and $d_v$ formulated in the following
(non-linear) minimization program (referred to as \MP).
 
 \begin{tbox}
 \begin{equation}
\begin{array}{ll@{}ll}
\text{minimize}  & F(d_u,d_v) &\\
\text{subject to}& \sum_{i \in [\card{\newM}]} d(u_i) + d(v_i) = \card{C}   \label{eq:MP}\\
                 & d(u_i),d(v_i) \in \bracket{\bMatch}~~~~~~~~ i=1 ,\ldots,\card{\newM}
\end{array}
\end{equation}
\end{tbox}
The constraint on each individual $d(u_i)$ and $d(v_i)$ is because $C$ is a $\bMatch$-matching.  The following lemma lower bounds the
value of the objective function in \MP.  


\begin{lemma}\label{lem:allocation}
  Let $\Fstar$ denote the optimal value of \MP; then, 
  \[\Fstar \ge \paren{p \cdot \card{C} - \card{\newM}} \cdot \eta -O(p_0)\cdot\OPT\]
  where $\eta := f({1 \over e})\cdot {1 \over e^2} \paren{1 - e^{-1}} > 0.07157$. 
\end{lemma}


The proof of the Lemma~\ref{lem:allocation} is technical, and we defer it to Section~\ref{sec:mp}. By Lemma~\ref{lem:allocation} and
Claim~\ref{clm:C-size} (the lower bound on $\card{C}$) the expected increment (over $\newM$) of the matching size is at least
\begin{align*}
  \Fstar \ge &\paren{p \cdot \card{C} - \card{\newM}} \cdot \eta -O(p_0)\cdot\OPT \\
  \ge  &\paren{p \cdot(2 \card{\BinM} - 2 \card{\newM}/p) - \card{\newM}} \cdot   \eta - O(p_0)\cdot\OPT \tag{By Claim~\ref{clm:C-size}}\\
  =  &\paren{ 2p \card{\BinM} - 3 \card{\newM}} \cdot \eta - O(p_0)\cdot\OPT\\
  =  &\paren{ 2p (\card{B} - \card{\BoutM}) - 3 (\card{M} + \outMs)} \cdot \eta - O(p_0)\cdot\OPT \\
  =  &\paren{ 2p \card{B} - 3\card{M} - 2p\card{\BoutM} - 3\outMs} \cdot \eta - O(p_0)\cdot\OPT \\
  = &\paren{ 2p (1 - O(p_0)){\OPT \over p} - 3 \paren{{1 \over 2} - \eps_0}\OPT - 2p\card{\BoutM} - 3\outMs} 
      \cdot  \eta - O(p_0)\OPT \\
  = &\paren{ \paren{ {1 \over 2} + 3\eps_0}{\OPT} - 2p\card{\BoutM} - 3\outMs} \cdot  \eta - O(p_0)\cdot\OPT
\end{align*}


Since the original matching $\newM$ is of size $(1/2 - \eps_0)\cdot\opt + \outMs$, the expected matching size in $M \cup B_p$, i.e.,
$\mu(M\cup B_p)$, is
\begin{align*}
  &\expect{\mu(M\cup B_p)} = \sum_{\outMs} \prob{\outMs} \expect{\mu{(M\cup B_p)} \mid \outMs} \tag{$\expect{X} = \sum_Y \prob{Y}\expect{X \mid
  Y}$} \\
\ge   &\sum_{\outMs} \prob{\outMs} (\card{\newM} + {\Fstar}) \\
\geq&  \paren{{1 \over 2} - \eps_0}\cdot \opt + \expect{\outMs} +
                                            \paren{ \paren{ {1 \over 2} + 3\eps_0}{\OPT} - 2p\card{\BoutM} - 3\expect{\outMs}}
 \cdot  \eta - O(p_0)\cdot\OPT\\
\ge&  \paren{{1 \over 2} + {\eta \over 2} - \eps_0  + 3\eta\eps_0}\cdot \opt + (1 - 3\eta)\expect{\outMs} - 2p\eta\card{\BoutM} -
     O(p_0)\cdot\OPT\\
\ge&  \paren{{1 \over 2} + {\eta \over 2} - \eps_0  + 3\eta\eps_0}\cdot \opt + (1 - 3\eta) (1 - O(p_0)){p\card{\BoutM} \over 3} -
     2p\eta\card{\BoutM} - O(p_0)\cdot\OPT \tag{By Equation~\ref{eq:bout}}\\
\ge&  \paren{{1 \over 2} + {\eta \over 2} - \eps_0  + 3\eta\eps_0}\cdot \opt + \paren{{1 \over 3} - 3\eta}p\card{\BoutM} - O(p_0)\cdot\OPT
\end{align*}
Since $\eta \approx 0.07157$, ${1 \over 3} - 3\eta > 0$, and we have 
\begin{align*}
  &  \paren{{1 \over 2} + {\eta \over 2} - \eps_0  + 3\eta\eps_0}\cdot \opt + \paren{{1 \over 3} - 3\eta}p\card{\BoutM} -
       O(p_0)\cdot\OPT\\
  \ge &  \paren{{1 \over 2} + {\eta \over 2} - \eps_0  + 3\eta\eps_0}\cdot \opt - O(p_0)\cdot\OPT\\
  > &~0.52 \cdot \OPT \tag{$\eps_0 = 0.02001$, $\eta > 0.07157$, and $p_0$ is sufficiently small.} \\  
   =  & (1/2 + \delta_0) \cdot \OPT
\end{align*}
completing the proof of Lemma~\ref{lem:b-matching-augmentation-raw}.

\subsection{Lower Bounding the Value of \MP}\label{sec:mp}

  \newcommand{\dustar}{\ensuremath{d^\star_u}} 
  \newcommand{\duistar}{\ensuremath{d^\star(u_i)}} 
  \newcommand{\duipstar}{\ensuremath{d^\star(u_{i'})}} 
  \newcommand{\duiOstar}{\ensuremath{d^\star(u_{i_1})}} 
  \newcommand{\duiTstar}{\ensuremath{d^\star(u_{i_2})}} 
  \newcommand{\dvstar}{\ensuremath{d^\star_v}} 
  \newcommand{\dvistar}{\ensuremath{d^\star(v_i)}} 
  \newcommand{\dviOstar}{\ensuremath{d^\star(v_{i_1})}} 
  \newcommand{\dviTstar}{\ensuremath{d^\star(v_{i_2})}}

In this section, we prove Lemma~\ref{lem:allocation}, i.e., the following inequality,
\begin{align*} 
\Fstar \paren{= \min F(d_u, d_v)} \ge (p \cdot \card{C} - \card{\newM}) \cdot \eta - O(p_0)\OPT 
\end{align*}
where $\eta := f({1 \over e})\cdot {1 \over e^2} \paren{1 - e^{-1}}$.  

Recall that 
  \begin{align*}
    F(d_u,d_v) = &\sum_i  \paren{1 - O(p_0)} \cdot  f({1 \over e})\cdot {p \over e^2} \paren{1 - e^{-p \cdot d(v_i)}} \cdot \max\set{d(u_i) - 1,0} \\
=  &  \paren{1 - O(p_0)} \cdot  f({1 \over e})\cdot {p \over e^2}  \sum_i\paren{1 - e^{-p \cdot d(v_i)}} \cdot \max\set{d(u_i) - 1,0} 
  \end{align*}
  Since the term $ \paren{1 - O(p_0)} \cdot  f({1 \over  e})\cdot {p \over e^2} $ is independent of $d_u$ and $d_v$, 
  \begin{align*}
   \arg\min_{d_u,d_v}  F =  \arg\min_{d_u, d_v} \sum_i \paren{1 - e^{-p \cdot d(v_i)}} \cdot \max\set{d(u_i) - 1,0} 
  \end{align*}

  Define $d(V) := \sum_id(v_i)$ and $d(U) := \sum_i d(u_i)$; then, $d(V) + d(U) = \card{C}$.  We need to prove that for any choice of $d(V)$
  and $d(U)$, the lemma statement holds.  First of all, we can assume $d(U) \ge \card{\newM}$: otherwise, since $d(V) \le \bMatch\card{\newM}$, we
  will have
  \begin{align*}
    \card{C} = d(V) + d(U) \le \bMatch\card{\newM} + \card{\newM} \le \paren{{1 \over p} + 1} \card{\newM}
  \end{align*}
  Therefore, for the target lower bound on $\Fstar$
  \begin{align*}
    &  (p \cdot \card{C} - \card{\newM})\cdot \eta - O(p_0)\OPT \\
   \le   &  \paren{p \paren{{1 \over p} + 1} \card{\newM}  - \card{\newM}} \cdot \eta - O(p_0)\OPT  \\
   \le   &  p \eta \card{\newM} - O(p_0)\OPT \\
   \le   &  p_0 \eta \card{\newM} - O(p_0)\OPT 
  \end{align*}
  which can be made negative by choosing the constant hidden in $O(p_0)$ to be $1$, proving Lemma~\ref{lem:allocation}.

  We further assume $d(U) - \card{\newM}$ is an integer multiple of $\bMatch -1$ and $d(V)$ is an integer multiple of $\bMatch$. This can be
  achieved by removing at most $1/p$ edges from $d(U)$ and $d(V)$ respectively. Since $F$ is monotonically increasing with any $d(u_i)$ or
  $d(v_i)$, removing edges from $d(U)$ and $d(V)$ can only make $\Fstar$ even smaller. Therefore, if we show that after removing these
  edges, the target lower bound on $\Fstar$ holds, then it definitely holds for the original $d(U)$ and $d(V)$. In the following, we fix any
  $d(U)$ and $d(V)$ where $d(U) - \card{\newM}$ is an integer multiple of $\bMatch -1$, $d(V)$ is an integer multiple of $\bMatch$, and
  $d(V) + d(U) \ge \card{C} - {2 \over p}$. We prove the following key property of $F(d_u, d_v)$.
  
  \begin{lemma}\label{lem:two-values}
    There exists $d_u$ and $d_v$ that minimizes $F(d_u, d_v)$ where any entry in $d_u$ is either $1$ or $\bMatch$ and any entry in $d_v$ is
    either $0$ or $\bMatch$.
  \end{lemma}
  
  We first show why Lemma~\ref{lem:two-values} implies the target lower bound on $\Fstar$, and then prove Lemma~\ref{lem:two-values}. Fix 
  $\duistar$ and $\dvistar$ that satisfy the property in Lemma~\ref{lem:two-values}.  Since every entry in $\dustar$ is either $1$ or
  $\bMatch$, the number of $1$'s in $\dustar$ is $x := \card{\newM} - (d(U) - \card{\newM})/( \bMatch - 1)$. Similarly, the number of $0$'s in
  $\dvstar$ is $y := \card{\newM} - d(V)/\bMatch$. Therefore, the number of edges in $\newM$ where $\dvistar = \duistar = \bMatch$ is at least
  \begin{align*}
    \card{\newM} - x - y = &\card{\newM} - \paren{\card{\newM} - {d(U) - \card{\newM} \over \bMatch - 1}} - \paren{ \card{\newM} - {d(V)
                             \over \bMatch}}\\
    = & {d(U) - \card{\newM} \over \bMatch - 1} -  \card{\newM} + {d(V) \over \bMatch}\\
    \ge &   {d(U) - \card{\newM} \over {1 \over p}} -  \card{\newM} + d(V)p \tag{$\bMatch -1 \le \bMatch \le {1 \over p}$}\\
    = & p \cdot d(U) + p \cdot d(V) - (1 + p)\card{\newM}\\
    \ge & p \paren{\card{C} - {2 \over p}} - (1 + p)\card{\newM} \tag{$d(V) + d(U) \ge \card{C} - {2 \over p}$.}\\
    = &p \card{C} - \card{\newM} - O(p_0)\OPT \tag{$\card{\newM} = O(\OPT)$}
  \end{align*}


  And just focusing on these $p \card{C} - \card{\newM} - O(p_0)\OPT$ edges, we have
  \begin{align*}
       \Fstar \ge &\paren{p \card{C} - \card{\newM} - O(p_0)\OPT} \cdot (1- O(p_0)) \cdot  f({1\over  e})\cdot {p \over e^2} \paren{1 - e^{-1}} \cdot (\bMatch -
                    1) \\
    \ge &\paren{p \card{C} - \card{\newM} - O(p_0)\OPT} \cdot (1- O(p_0)) \cdot  f({1\over  e})\cdot {p \over e^2} \paren{1 - e^{-1}} \cdot
          {1 \over p}   \tag{$\bMatch -1 \ge (1 - O(p_0)){1 \over p}$}\\
    \ge &\paren{p \card{C} - \card{\newM} - O(p_0)\OPT} \cdot (1- O(p_0)) \cdot  \eta 
             \\
    \ge &\paren{p \card{C} - \card{\newM} } \cdot  \eta  - O(p_0)\OPT \tag{$\card{\newM} = O(\OPT)$, $p \card{C} \le p\card{\BinM}= O(\OPT)$}\\
  \end{align*}
  which proves Lemma~\ref{lem:allocation}.



  We now prove Lemma~\ref{lem:two-values} which will complete the proof.  
  \begin{proof}[Proof of Lemma~\ref{lem:two-values}]
    Fix any allocation $\dustar$ and $\dvstar$ that minimizes $F(d_u, d_v)$. We will show that first, there exists a sequence of locally
    reallocating the values (i.e., degrees) in $\dustar$ without changing the value of $F(d_u, d_v)$ such that at the end, every entry in
    $\dustar$ is either $1$ or $\bMatch$. After changing the vector $\dustar$, we then show that there exists a sequence of locally
    reallocating the values in $\dvstar$ without changing the value of $F(d_u, d_v)$ such that at the end, every entry in $\dvstar$ is
    either $0$ or $\bMatch$.

    We first explain how to change $\dustar$.  To simplify the presentation, we define $q_i = {1 - e^{-p \cdot \dvistar}}$ and the target
    expression becomes
    \begin{align}
      \sum_i \paren{1 - e^{-p \cdot \dvistar}} \cdot \max\set{\duistar - 1,0}  =  \sum_i q_i \cdot \max\set{\duistar -1, 0}   \label{eq:has-q}
   \end{align}
   
   Recall that $\duistar$ satisfies $\duistar \in \bracket{\bMatch}$ (in \MP) (and hence $q_i \ge 0$) and $\sum_i \duistar = d(U)$.  First
   of all, if there exists some $i_1$ where $\duiOstar = 0$, then since $d(U) \ge \card{\newM}$, there must exist an index $i_2$ where
   $\duiTstar \ge 2$. Then, we can shift one degree from $\duiTstar$ to $\duiOstar$ and after the shift, $(a)$ $\max\set{\duiOstar -1, 0}$
   remains $0$ and hence $q_{i_1}\max\set{\duistar -1, 0}$ remains $0$, and $(b)$ $\max\set{\duiTstar -1, 0}$ decreases and hence
   $q_{i_1} \max\set{\duiTstar -1, 0}$ does not increase. Therefore, $F(d_u, d_v)$ does not increase after the shift, and from now on, we
   have $\duistar \ge 1$ for all $i \in [\card{\newM}]$.  To proceed, we need the following property of $\duistar$.

\newcommand{\diff}{\Delta}

  \begin{claim}\label{clm:duistar}
    For any pair of indices $i_1,i_2$ where $q_{i_1} > q_{i_2}$, either  $\duiTstar = \bMatch$ or $\duiOstar \le 1$.
  \end{claim}
  \begin{proof}
    Suppose not. We have $\duiTstar < \bMatch$ and $\duiOstar > 1$, for some $i_1$ and $i_2$. We can shift one degree from $\duiOstar$ to
    $\duiTstar$ and still get a valid allocation. In the following, we show that this new allocation achieves a smaller value of $F$, which
    contradicts to the optimality of $\duistar$.

    Since shifting from $\duiOstar$ to $\duiTstar$ only changes the degrees for $u_{i_1}$ and $u_{i_2}$, it suffices for us to prove that  
    \begin{align*}
     \diff :=  &(q_{i_2} \max\set{\duiTstar -1, 0} + q_{i_1} \max\set{\duiOstar-1,0}) \\
      &- (q_{i_2} \max\set{\duiTstar, 0}  + q_{i_1} \max\set{\duiOstar - 2, 0}) > 0
    \end{align*}
    Since $\duiOstar \ge 2$, $\max\set{\duiOstar - 1, 0} = \duiOstar - 1$, $\max\set{\duiOstar - 2, 0} = \duiOstar - 2$. In addition,
    since $\duiTstar \ge 0$, $ \max\set{\duiTstar, 0} = \duiTstar $. We have 
    \begin{align*}
     \diff =  &(q_{i_2} \max\set{\duiTstar -1, 0} + q_{i_1} (\duiOstar-1)) - (q_{i_2} \duiTstar + q_{i_1} (\duiOstar - 2))\\
                 = & q_{i_2} \paren{\max\set{\duiTstar -1, 0} - \duiTstar} +q_{i_1}\\
      \ge & q_{i_2} (\duiTstar -1 - \duiTstar) +q_{i_1} \\
      = & q_{i_1} - q_{i_2}.
    \end{align*}
    Since we have $q_{i_1} > q_{i_2}$, the value of $F$ decreases after the shifting according to Eq~\ref{eq:has-q}, a contradiction.
  \end{proof}

  We use Claim~\ref{clm:duistar} to prove the correctness of the following sequence of reallocation of $\dustar$.  Now, as long as there
  exists an index $i_1$, where $\duiOstar \in (1, \bMatch)$, since $d(U) - \card{\newM}$ is an integer multiple of $(\bMatch - 1)$, there must
  exists some index $i_2$ where $\duiTstar \in (1, \bMatch)$ (recall that $\duistar \ge 1$), and we will shift the values between
  $\duiOstar$ and $\dviTstar$ such that one of them becomes either $1$ or $\bMatch$ and both of them are still at least $1$ (it is easy to
  see this is always possible). First of all, every step of this reallocation reduces the number of vertices with
  $\duistar \in (1, \bMatch)$, and hence it will terminate. To see that this process never changes $F(d_u, d_v)$, $(a)$ it cannot be that
  $q_{i_1} \neq q_{i_2}$, since otherwise the indices $i_1$ and $i_2$ will contradict Claim~\ref{clm:duistar}, and $(b)$ if
  $q_{i_1} = q_{i_2}$, shifting the allocation between $i_1$ to $i_2$ will not change $F(d_u, d_v)$.

  Therefore, we can focus on the case where the entries of $\duistar$ are either $1$ or $\bMatch$.  We now consider $\dvistar$. Recall that
  $d(V) = \sum_i d(v_i)$ and $d(V)$ is an integer multiple of $\bMatch$. The target expression can be written as
  \begin{align*}
     &\sum_i  \paren{1 - e^{-p \cdot \dvistar}} \cdot \max\set{\duistar - 1, 0} \\
    =  &  \sum_{i:~ \duistar = \bMatch}  \paren{1 - e^{-p \cdot \dvistar}} \paren{\bMatch - 1} + \sum_{i:~ \duistar = 1}  \paren{1 - e^{-p
         \cdot \dvistar}} \cdot 0 \\
    =  &  \sum_{i:~ \duistar = \bMatch}  \paren{1 - e^{-p \cdot \dvistar}} 
  \end{align*}
  Since $F$ is monotonically increasing when any $d(u_i)$ increases, for the indices $i$ where $\duistar = 1$, ideally, one should allocate
  as many degrees to $\dvistar$ as possible, i.e., $\dvistar = \bMatch$. However, it might be the case that $d(V)$ cannot supply $\bMatch$
  degree for all $i$ where $\duistar = 1$. But in this case, we are done since reallocating between different $\dvistar$ where
  $\duistar = 1$ does not change the $F$ (in fact, $F$ is always $0$), and we can shift them such that we have as many $\dvistar = \bMatch$
  as possible and leave the rest equal to $0$.

  In the following, we assume $d(V)$ can supply $\bMatch$ degree for all $i$ with $\duistar = 1$, and hence $\dvistar = \bMatch$ whenever
  $\duistar = 1$. It suffices to only focus on $\dvistar$ where $\duistar = \bMatch$. We need the following property of $\dvistar$ to
  complete the argument.
  
  \begin{claim}\label{clm:dvistar}
    For any pair of  indices $i_1$ and $i_2$ such that $\duiOstar = \duiTstar = \bMatch$, we have that either $\min\set{\dviOstar, \dviTstar} = 0$ or
    $\max\set{\dviOstar, \dviTstar} = \bMatch$.
  \end{claim}
  \begin{proof}
    Suppose not. Then, for some $i_1$ and $i_2$, we have $\min\set{\dviOstar, \dviTstar} > 0$, and also $\max\set{\dviOstar, \dviTstar} <
    \bMatch$. Without lose of generality, assume $1 \le \dviOstar \le \dviTstar \le \bMatch -1$. Then, shifting one degree from $\dviOstar$ to
    $\dviTstar$ leads to a valid allocation, and we prove in the following that the new allocation decreases the objective function which
    contradicts the optimality of $\dvstar$.
    
    Since only the indices $i_1$ and $i_2$ are affected, it suffices for us to prove that
    \begin{align*}
      \diff := \paren{1 - e^{-p \cdot \dviOstar}} + \paren{1 - e^{-p \cdot \dviTstar}}  -   \paren{1 - e^{-p \cdot (\dviOstar - 1)}} + \paren{1
      - e^{-p \cdot (\dviTstar + 1)}} > 0
    \end{align*}
    We have 
    \begin{align*}
       \diff  \geq e^{-p \cdot \dviOstar} (e^p - 1) + e^{-p \cdot \dviTstar}(e^{-p} - 1)
    \end{align*}
    Since $\dviOstar \le \dviTstar$, $e^{-p \cdot \dviOstar} \ge e^{-p \cdot \dviTstar}$. We further have
    \begin{align*}
      e^{-p \cdot \dviOstar} (e^p - 1) + e^{-p \cdot \dviTstar}(e^{-p} - 1) &\ge  e^{-p \cdot \dviTstar} (e^p - 1) + e^{-p \cdot \dviTstar}(e^{-p} - 1)\\
      &\ge e^{-p \cdot \dviTstar} (e^p - 1 + e^{-p} - 1) \\
      &>  e^{-p \cdot \dviTstar} (2 \sqrt{e^p \cdot e^{-p}} - 2)  \\
      &= 0
    \end{align*}
    where the strict inequality is true since the two terms can only be equal when $e^{p} = e^{-p}$ which does not happen for $p > 0$.
  \end{proof}
  Using Claim~\ref{clm:dvistar}, we can now show that any $\dvistar$ where $\duistar = \bMatch$, must be either $0$ or $\bMatch$. Suppose
  not. If $\dviOstar \in (0, \bMatch)$, then since $d(V)$ is an integer multiple of $\bMatch$, there must exists some other index $i_2$
  where $\dviTstar \in (0, \bMatch)$; hence, $0 < \min\set{\dviOstar, \dviTstar} < \max\set{\dviOstar, \dviTstar} < \bMatch$, a
  contradiction to Claim~\ref{clm:dvistar}.
\end{proof} 
\end{proof}

\newcommand{\LP}{\ensuremath{\textnormal{LP-(\ref{eq:lp})}}\xspace}
\section{An Algorithm for Large Values of $p$}\label{app:large-p}

In this section, we provide an algorithm, namely, Algorithm~\ref{alg:p-approximation}, with approximation ratio strictly better than $1/2$ when $p$ is bounded away from zero. In particular, this algorithm computes a 
matching of size $\paren{1/2 + \Theta(p^2)}\cdot \opt$. Algorithm~\ref{alg:p-approximation} is required to handle the case when $p$ is not small enough for
Algorithm~\ref{alg:small-p} to perform well. Using a combination of both of these algorithms, we
can prove the second part of Theorem~\ref{thm:main}.

Let $p_0$ be any fixed constant independent of $n$, $\delta = \frac{p^2}{4}$, and $\eps =\frac{\constp^2}{10^{4}}$. 
The new algorithm (i.e, Algorithm~\ref{alg:p-approximation}) is similar to Algorithm~\ref{alg:small-p} with the only difference being that instead of a $\bMatch$-matching, here, we simply pick a single maximum
matching in $G$. Our algorithm is stated as Algorithm~\ref{alg:p-approximation}.

\begin{algorithm}[h!]
\textnormal
\SetAlgoNoLine
\KwIn{A graph $G(V,E)$ and an edge realization probability  $p_0 \leq p < 1$.}
\KwOut{A subgraph $H(V,Q)$ of $G(V,E)$.}
\begin{enumerate}
\item Let $M$ be a maximum matching in $G$. 
\item Let $(M_1, M_2, \ldots, M_R) := \basicalg(G(V,E\setminus M),\eps)$ (recall that $\eps = \frac{\constp^2}{10^{4}}$), and
  $E_{MC} = M_1 \cup \ldots \cup M_R$.
\item Return $H(V,Q)$ where $Q:= M \cup E_{MC}$. 
\end{enumerate}
\caption{A $\paren{0.5+\Theta(p^2)}$-Approximation Algorithm for Stochastic Matching}
\label{alg:p-approximation}
\end{algorithm}

The following lemma proves the approximation ratio of Algorithm~\ref{alg:p-approximation}. 

\begin{lemma}\label{lem:p-approximation}
  For any constant $\constp > 0$, any realization probability $p \geq \constp$, and any graph $G(V,E)$ the expected maximum matching size
  in the graph $H$ computed by Algorithm~\ref{alg:p-approximation} is at least
  $\paren{\frac{1}{2}+\frac{p^2}{4} - \frac{\constp^2}{10^{4}}} \cdot {\OPT(G)}$.
\end{lemma}

Before proving Lemma~\ref{lem:p-approximation}, we show how to combine Algorithm~\ref{alg:small-p} and Algorithm~\ref{alg:p-approximation}
to prove Part~(\ref{part:main2}) of Theorem~\ref{thm:main}.
\begin{proof}[Proof of Theorem~\ref{thm:main}, Part~(\ref{part:main2})]
  Let $p_0$ be the constant such that Algorithm~\ref{alg:small-p} achieves an approximation ratio of $0.52$ for any $p \leq p_0$. The algorithm
  for Part~(\ref{part:main2}) is simply as follows. If the realization probability $p \leq p_0$, run Algorithm~\ref{alg:small-p} and
  otherwise run Algorithm~\ref{alg:p-approximation}.  By Lemma~\ref{lem:p-approximation}, the approximation ratio of this algorithm is
  $\min\set{0.52 , \frac{1}{2} + \frac{p^2}{4} - \frac{\constp^2}{10^4}} = 0.5 + \delta_0$ for some absolute constant $\delta_0$ (since
  $p_0$ is an absolute constant and $p \geq p_0$).
	
  We note that by optimizing the choice of $p_0$ and a more careful analysis of Algorithm~\ref{alg:small-p} (to account for many constants
  involved), one can bound the value of $\delta_0 \approx 0.001$. We omit the tedious details of this calculation as it is not the main
  contribution of this paper.
\end{proof}

We now prove Lemma~\ref{lem:p-approximation}. 

\begin{proof}[Proof of Lemma~\ref{lem:p-approximation}]
  Recall that $\opt$ (resp. $\ALG$) is the expected maximum matching size in a realization $G_p$ of $G$ (resp. a realization $H_p$ of $H$).
	
  Firstly, by Claim~\ref{clm:small-L}, with the parameters $\eps$, $\delta$, and $X = M$, we have that if $\card{M_R}$ in $E_{MC}$ is
  smaller than $({1 \over 2} - \delta)\cdot\opt$, then the expected matching size in $G(V,Q)$ is at least
  $({1 \over 2} + \delta - \eps)\cdot\opt = ({1 \over 2} + {p^2 \over 4} - \frac{\constp^2}{10^{4}}) \cdot \opt$, which proves the lemma. We now
  consider the case where $\card{M_R} \geq ({1 \over 2} - \delta)\cdot\opt$.
  
  Let $M'$ be the random variable denotes a maximum matching in a realization of $E_{MC}$ (breaking tie arbitrarily). By
  Lemma~\ref{lem:basic-alg}, w.h.p., $\card{M'} \ge (1-\eps)\card{M_R} \geq (\frac{1}{2} - \delta - \eps)\cdot\opt$.  For
  simplicity, in the following, we always assume this event happens\footnote{This assumption can be removed while losing a negligible factor
    of $o(1)$ in the size of final matching.}  and further remove any extra edges in $M'$ so that
  $\card{M'} = (\frac{1}{2} - \delta - \eps)\cdot\opt$.  We now use the matching $M$ chosen in the first step of the algorithm (which is
  a maximum matching of $G$) to augment the matching $M'$. We should point out that at this point, $M'$ refers to a realized matching, while
  $M$ is still a random variable (independent of $M'$ since $M$ and $E_{MC}$ are edge-disjoint).
	
  Let $\alpha_1,\alpha_3$ and $\alpha_{\geq 5}$ denote, respectively, the number of augmenting paths (w.r.t. $M'$) of length $1$, $3$, and
  at least $5$ in $M \bigtriangleup M'$.  We have the following claim. The proof uses standard facts about the augmenting paths (see,
  e.g.,~\cite{HopcroftK73}).
	\begin{claim}\label{clm:augmenting-paths}
		For $\alpha_1,\alpha_3$, and $\alpha_{\geq 5}$, defined as above: 
		\begin{align*}
			\alpha_3 + 2\alpha_{\geq 5} &\leq \card{M'} \\
			\alpha_1 + \alpha_3 + \alpha_{\geq 5} &= \card{M} - \card{M'}
		\end{align*}
              \end{claim}
              \begin{proof}
                Any augmenting path of length $3$ has one edge in $M'$ and any augmenting path of length at least $5$ has at least two edges
                in $M'$. Since the augmenting paths are edge disjoints, the first inequality follows. The second inequality follows from the
                fact that $M$ is a maximum matching in $G$ and each augmenting path in $M \bigtriangleup M'$ increases the size of $M'$ by
                $1$.
              \end{proof}
              As stated earlier, each edge in $M$ is realized w.p. $p$ (independent of the choice of $M'$).  Since an augmenting path of
              length $1$ (resp. of length $3$) realizes in $M' \bigtriangleup M_p$ w.p. $p$ (resp. $p^2$), we have that the expected number
              of times that $M'$ can be augmented using realized edges of $M$ is at least $\alpha_1 p + \alpha_3 p^2$, implying that the
              final matching size is at least $(\frac{1}{2} - \delta - \eps)\cdot\opt + \alpha_1 p + \alpha_3 p^2 $ in
              expectation. Combining this with Claim~\ref{clm:augmenting-paths}, the minimum size of the output matching we obtain can be
              formulated as the following linear program (denoted by \LP):
	 \begin{tbox}
 \begin{equation}
\begin{array}{ll@{}ll}
\text{minimize}  & \alpha_1 p + \alpha_3 p^2  &\\
\text{subject to}& \alpha_3 + 2\alpha_{\geq 5} \leq (\frac{1}{2} - \delta)\opt - \eps\cdot\opt    & \label{eq:lp}\\
			&     \alpha_1 + \alpha_3 + \alpha_{\geq 5} \geq (\frac{1}{2} + \delta)\opt + \eps\cdot\opt & \\
                 & \alpha_1,\alpha_3,\alpha_{\geq 5} \geq 0 
\end{array}
\end{equation}
\end{tbox}
	
	where in the second constraint, we use the fact that $M$ is a maximum matching in $G$ and hence $\card{M} \geq \opt$. We have the following claim. 
	
	\begin{claim}\label{clm:lp-minimizer}
		The minimum value of $\LP$ is at least ${\frac{p^2}{2}} \cdot \opt$.
	\end{claim}
	\begin{proof}
		The two constraints of \LP imply that, 
		\begin{align}
			2\alpha_1 + \alpha_3 \geq \paren{\frac{1}{2} + 3\delta + 3\eps} \cdot \opt  \label{eq:constraint-both}
		\end{align}
		
		Suppose we want to minimize $\alpha_1 p + \alpha_3 p^2$ subject to the constraint in Eq~(\ref{eq:constraint-both}) (this is clearly a lower bound for the value of \LP). In this case, 
		since the contribution of $\alpha_3$ to the objective value is $p$ times the contribution of $\alpha_1$, while its contribution to the constraint is $\frac{1}{2}$ times the contribution of $\alpha_1$, it is 
		straightforward to verify that for $p \leq 1/2$, there is an optimal solution with $\alpha_1 = 0$, and for $p > 1/2$, there is an optimal solution with $\alpha_3 = 0$. We can now compute the value of 
		solution in each case: 
		
		\textbf{\bm{$p \leq \frac{1}{2}$} case.} In this case $\alpha_1 = 0$ and $\alpha_3 = \paren{\frac{1}{2} + 3\delta + 3\eps} \cdot \opt$ minimizes $\alpha_1 p + \alpha_3 p^2$. Hence, the objective value is 
		\begin{align*}
			\alpha_3 \cdot p^2 =  \paren{\frac{1}{2} + 3\delta + 3\eps} \cdot \opt \cdot p^2 \geq {\frac{p^2}{2}} \cdot \opt
		\end{align*}
		
		\textbf{\bm{$p > \frac{1}{2}$} case.} In this case $\alpha_1 = \paren{\frac{1}{4} + \frac{3}{2}\delta + \frac{3}{2} \eps} \cdot \opt$ and $\alpha_3 = 0$ minimizes $\alpha_1 p + \alpha_3 p^2$. Hence, the objective value is
		\begin{align*}
			\alpha_1 \cdot p &=  \paren{\frac{1}{4} + \frac{3}{2}\delta + \frac{3}{2}\eps}p \cdot \opt \geq \paren{\frac{1}{4} + \frac{3}{2}\delta}p \cdot \opt \\
			&= \paren{\frac{1}{4} + \frac{3p^2}{8}}p \cdot \opt  \geq {\frac{p^2}{2}} \cdot \opt \tag{$\delta = \frac{p^2}{4}$
                          and $\frac{1}{4} + \frac{3p^2}{8} \ge 2\sqrt{\frac{1}{4} \cdot \frac{3p^2}{8}} \ge \frac{p}{2}$}
		\end{align*}
		The claim now follows since in above calculation we \emph{relaxed} constraints of \LP to the constraint in Eq~(\ref{eq:constraint-both}).
	\end{proof}
	
	By plugging in the bound from Claim~\ref{clm:lp-minimizer}, we obtain that the final matching size is at least: 
	\begin{align*}
		\frac{\opt}{2}  - \delta\cdot\opt - \eps \cdot\opt + \alpha_1 p + \alpha_3 p^2  &\geq  \paren{{1 \over 2}  - \delta  -\eps +
                                                                                                  \frac{p^2}{2}} \cdot \opt  \\
		&= \paren{\frac{1}{2}+\frac{p^2}{4} - \frac{\constp^2}{10^4} }\cdot \opt 
	\end{align*}
	by plugging in $\delta = \frac{p^2}{4}$ and $\eps = \frac{\constp^2}{10^4}$. 		
\end{proof}

\section{Concluding Remarks and Open Problems}\label{sec:conc}
We presented the first non-adaptive algorithm for stochastic matching with an approximation ratio that is strictly better than half. In
particular, we showed that any graph $G$ has a subgraph $H$ with maximum degree $O(\frac{\log{(1/p)}}{p})$ such that the ratio of expected
size of a maximum matching in realizations of $H$ and $G$ is at least $0.52$ when $p$ is sufficiently small, i.e., case of vanishing probabilities, and $0.5+\delta_0$ (for an
absolute constant $\delta_0 > 0$) for any $p \in (0,1)$.

A main open problem is to determine the best approximation ratio achievable by a non-adaptive algorithm. In particular, can non-adaptive
algorithms qualitatively match the performance of adaptive algorithms by achieving a $(1-\eps)$-approximation for any $\eps > 0$ using a
subgraph with maximum degree $f(\eps,p)$ for some function $f$? In the following, we mention some potential directions towards resolving this problem. 

\paragraph{A barrier to obtaining a $(1-\eps)$-approximation.} We briefly explain here a barrier to a $(1-\eps)$-approximation algorithm that was noted in~\cite{ECpaper}. 
 It was shown in~\cite{ECpaper} that any (non-adaptive) $(1-\eps)$-approximation algorithm for stochastic matching needs to solve the following problem. 

\begin{problem*}[\!\cite{ECpaper}]
  Suppose you are given a bipartite graph $G(L, R, E)$ ($\card{L} = \card{R} = n$) with the property that the expected maximum matching size
  between two \emph{uniformly at random chosen} subsets $A \subseteq L$ and $B \subseteq R$ with $\card{A} = \card{B} = n/3$, is
  $n/3 - o(n)$. The goal is to compute a subgraph $H(L,R,Q)$ with max-degree of $O(1)$, such that the expected size of a maximum matching
  between two randomly chosen subsets $A$ and $B$ is $\Omega(n)$.
\end{problem*}

For the harder problem in which the two subsets $A$ and $B$ are chosen \emph{adversarially}, it is known that there exist graphs (in
particular, a \rs graph; see, e.g.~\cite{AlonMS12,FischerLNRRS02}) that admit no such sparse subgraph $H$ (see~\cite{ECpaper} for more
details).  However, in the stochastic matching application, our interest is in \emph{randomly} chosen subsets $A$ and $B$, and it is not
known if there are instances such that the random set version of the problem is hard.


\paragraph{A direct application of $b$-matching lemma.}  There is another possible way of utilizing the $b$-matching lemma. In
Lemma~\ref{lem:outside-edges}, we showed that for any $\frac{1}{p}$-matching $B$ of size $\frac{\opt}{p}$, the expected maximum matching
size of a realization of $B$ is at least $\frac{\opt}{3}$. In fact, using a more careful analysis, we can improve this bound to
$\approx 0.4 \cdot \opt$. This, together with our $b$-matching lemma, immediately implies a simple $0.4$-approximation algorithm for
stochastic matching. However, it is not clear to us whether this bound can be significantly improved to get a matching of size strictly
more than $\frac{\opt}{2}$.  It is worth mentioning that using a result of Karp and Sipser~\cite{KarpS81} on \emph{sparse random graphs}
(see also~\cite{AronsonFP98}, Theorem~4), one can show that if the $\frac{1}{p}$-matching itself is chosen \emph{randomly}, then its
realizations contain a matching of size $\approx 0.56 \cdot \opt$ in expectation. However, this result relies heavily on the fact that the
original graph (in our case a realization of a random $\frac{1}{p}$-matching) is chosen randomly, and it seems unlikely that a similar
result holds for an \emph{adversarially} chosen $\frac{1}{p}$-matching.

\bibliographystyle{abbrv}
\bibliography{general}

\clearpage
\appendix

\section{Omitted Proofs from Section~3}\label{app:prelim}

\subsection{Proof of Proposition~\ref{prop:upper-exp}}\label{app:upper-exp}
\begin{proof}
	We first have $f(x)$ is \emph{monotonically decreasing}, since, 
	\begin{align*}
		\frac{d f}{d x} = \frac{e^{-x} \cdot x - 1+e^{-x}}{x^2} = \frac{(x+1)\cdot e^{-x} - 1}{x^2} \leq \frac{e^{x}\cdot e^{-x} - 1}{x^2} = 0 
	\end{align*}
	where we used the inequality $(1+x) \leq e^{x}$. 
	
  Consequently, since $x \le c$, 
  \begin{align*}
    f(c) \le f(x) = {1 - e^{-x} \over x}
  \end{align*}
  which implies $e^{-x} \le 1 - f(c) \cdot x$.
\end{proof}

\subsection{Proof of Proposition~\ref{prop:upper-exp-2}}\label{app:upper-exp-2}
\begin{proof}
  We first exam the equivalent conditions for the target inequality.
  \begin{align*}
    &                             &(1-x)^{{1\over x}} \ge {1 - x \over e} \\ 
    &\Longleftrightarrow &     (1-x)^{{1\over x} - 1} \ge {1 \over e} \\
    &\Longleftrightarrow &     ({1\over x} - 1)\ln(1-x) \ge -1 \tag{by taking natural log of both sides}
  \end{align*}
  Now, since $\ln(1 - x) \ge -x - {x^2 \over 2} - {x^3 \over 2}$ when $x \in (0,0.43]$. We have 
  \begin{align*}
    ({1\over x} - 1)\ln(1-x) &\ge ({1\over x} - 1) (-x - {x^2 \over 2} - {x^3 \over 2}) \tag{since $({1\over x} - 1) > 0$} \\
    & = -1 - {x \over 2} - {x^2 \over 2} + x + {x^2 \over 2} + {x^3 \over 2} \\
    & = - 1 + {x \over 2} + {x^3 \over 2} \\
    & \ge -1
  \end{align*}
  which completes the proof.
\end{proof}

\newcommand{\cstar}{\ensuremath{c^{\star}}}

\section{The Optimality of the $b$-Matching Lemma} \label{app:b-limitation}

In this section, we establish that our $b$-matching lemma is essentially optimal in the sense that it is impossible to find a $b$-matching
with at least $b\cdot \opt(G)$ edge for $b$ much larger than $1/p$. In particular, we show that,

\begin{claim*}
  For any constant $0 < p < 1$, there exist bipartite graphs $G$ where $G_p$ has a matching of size $n-o(n)$ in expectation, but for any
  $b \geq {2 \over p}$, there is no $b$-matching in $G$ with (at least) $b\cdot 0.99n$ edges; here $n$ is the number of
  vertices on each side of $G$.
\end{claim*}
\begin{proof}
  For any integer $N$, let $\FG_{N,\frac{1}{N}}$ be the family of bipartite random graphs with $N$ vertices on each side and probability of picking each edge being $1/N$. 
  Let $\cstar \in (0,1)$ such that any bipartite graph sampled from $\FG_{N,\frac{1}{N}}$ has a matching of size at least $\cstar \cdot N$ w.p. $1-o(1)$. 
  By a result of Karp and Sipser~\cite{KarpS81} on \emph{sparse random graphs} (see also~\cite{AronsonFP98}, Theorem~4), we have $\cstar \approx 0.56$.
   
  Consider bipartite graphs $G(L,R,E)$ where the vertices in $L$ consists of two disjoint sets $L_1$ and $L_2$ with $\card{L_1} = N$ and
  $\card{L_2} = (1-\cstar) \cdot N$ for parameter $N = \frac{n}{2-\cstar}$. Similarly, $R$ contains two sets $R_1$ and $R_2$ with $\card{R_1} = N$
  and $\card{R_2} = (1-\cstar)\cdot N$.
  
  The set of edges in $G$ can be partitioned into two parts. First, there is a complete bipartite graph between $L_1$ and $R_2$, and a complete bipartite graph
  between $L_2$ and $R_1$. Second, there is a sparse graph between $L_1$ and $R_1$ defined through the following random process: each edge
  between $L_1$ and $R_1$ is independently chosen w.p. ${1 \over pN}$.

  In the following, we show that for a graph $G$ created through the above process, w.p. $1 - o(1)$, $G_p$ has a matching of size $n-o(n)$
  in expectation, and w.p. $1 - o(1)$, there is no $b$-matching in $G$ with $b\cdot 0.99n$ edges, for $b \geq {2 \over p}$. Hence, by applying a
  union bound, the above process find a graph with both properties w.p.  $1 - o(1)$, proving the claim.
  
  To see that $G_p$ has a matching of size $n-o(n)$ in expectation, we realize the edges in $G$ in two steps: first realize the edges
  between $L_1$ and $R_1$, and then the other edges (i.e., the two complete graphs between $L_1, R_2$ and between $L_2, R_1$,
  respectively). For the subgraph between $L_1$ and $R_1$, notice that each edge between $L_1$ and $R_1$ is realized w.p.
  ${1 \over pN} \cdot p = {1 \over N}$ (chosen w.p. ${1 \over pN}$ in the above process and realize w.p. $p$). Since
  $\card{L_1} = \card{R_1} = N$, the subgraph between $L_1$ and $R_1$ is sampled from $\FG_{N,\frac{1}{N}}$ and hence w.p. $1-o(1)$, 
  there is a matching of size $\cstar N$ between $L_1$ and $R_1$. Now for the remaining $(1-\cstar) N$ unmatched vertices in $L_1$ (resp. in $R_1$), since there is a complete graph between
  $L_1$ and $R_2$ (resp. $R_1$ and $L_2$), w.p.  $1 - o(1)$, a perfect matching realizes between the unmatched vertices in $L_1$ and vertices in $R_2$
  (resp. between $R_1$ and $L_2$). We conclude that any realization $G_p$ has a perfect matching w.p. $1-o(1)$ and hence the expected maximum matching
  size in $G_p$ is at least $(1 - o(1)) n + o(1) \cdot 0 = n - o(n)$.

  It remains to show that w.p. $1 - o(1)$, $G$ has no $b$-matching with $b\cdot 0.99n$ edges for $b\geq{2 \over p}$. For any $b$-matching in $G$,
  the number of edges incident on $L_2$ and $R_2$ is at most $b \cdot (\card{L_2} + \card{R_2}) = 2 b N/(1-\cstar)$. The remaining edges of this
  $b$-matching must be between $L_1$ and $R_1$. Each edge between $L_1$ and $R_2$ is chosen w.p. ${1 \over pN}$, and there are $N^2$
  possible edges between $L_1$ and $R_1$. By Chernoff bound, w.p. $1 - o(1)$, the number of realized edges between $L_1$ and $R_1$ is at
  most $(1 +o(1)) {N \over p}$. Therefore, the total number of edges of any $b$-matching in
  $G$ is at most
  \begin{align*}
    2 b (1-\cstar) N + (1 + o(1)){N \over p} &= b\cdot n - \paren{\cstar b  \cdot N - \frac{N}{p}} + o(n) \tag{$(2-\cstar) \cdot N = n$}\\
    & \leq b\cdot n - \paren{0.56 b N - 0.5bN} + o(n) \tag{$b \geq 2/p$ and hence $1/p \leq b/2$; $\cstar \approx 0.56$}\\
    & < b\cdot 0.99 n 
    \end{align*}
\end{proof}

\end{document}